%% file: main.tex
\documentclass[sigplan,screen]{acmart}
\usepackage{algorithm, algpseudocode}
\usepackage{xparse}
\usepackage{graphicx,subcaption}
\usepackage{color}

\usepackage{soul}

\makeatletter
\NewDocumentCommand{\LeftComment}{s m}{%
  \Statex \IfBooleanF{#1}{\hspace*{\ALG@thistlm}}\(\triangleright\) #2}
\makeatother
\AtBeginDocument{%
  }

\copyrightyear{2026}
\acmYear{2026}
\setcopyright{cc}
\setcctype{by}
\acmConference[PPoPP '26]{Proceedings of the 31st ACM SIGPLAN Annual Symposium on Principles and Practice of Parallel Programming}{January 31-February 04, 2026}{Sydney, NSW, Australia}
\acmBooktitle{Proceedings of the 31st ACM SIGPLAN Annual Symposium on Principles and Practice of Parallel Programming (PPoPP '26), January 31-February 04, 2026, Sydney, NSW, Australia}
\acmPrice{}
\acmDOI{10.1145/3774934.3786458}
\acmISBN{979-8-4007-2310-0/2026/01}

\graphicspath{ {figs/} }

\begin{document}

\title{
Sharded Elimination and Combining for Highly-Efficient Concurrent Stacks
}


\author{Ajay Singh}
\affiliation{%
  \institution{FORTH ICS}
  \streetaddress{}
  \city{}
  \country{}}
  \orcid{0000-0001-6534-8137}
\email{ajay.singh1@uwaterloo.ca}

\author{Nikos Metaxakis}
\affiliation{%
  \institution{FORTH ICS}
  \streetaddress{}
  \city{}
  \country{}}
  \orcid{0009-0000-6731-474X}
\email{csdp1437@csd.uoc.gr}

\author{Panagiota Fatourou}
\affiliation{%
  \institution{FORTH ICS}
  \streetaddress{}
  \city{}
  \country{}}
  \orcid{0000-0002-6265-6895}
\email{faturu@ics.forth.gr}


\input{macro}

\begin{abstract}
We present a new blocking linearizable stack implementation which utilizes {\em sharding} and fetch\&increment
to achieve significantly better performance than all existing concurrent stacks. 
The proposed implementation is based on a novel elimination mechanism and a
new combining approach that are efficiently blended to gain high performance.
Our implementation results in enhanced parallelism and low contention when accessing the shared stack.
Experiments show that the proposed stack implementation outperforms all existing 
concurrent stacks by up to 2$\times$ in most workloads. It is particularly efficient in  systems supporting a large number of threads and in high contention scenarios.
\end{abstract}



\keywords{concurrent stacks, elimination, software combining, concurrent data structures}


\maketitle

\input{sections/introduction}

\input{sections/relatedwork}

\input{sections/DCEalgorithm}

\input{sections/shortRecandCorrectness}
\input{sections/experiments}


\begin{acks}
Supported by the Greek Ministry of Education, Religious Affairs and Sports call SUB 1.1 -- Research Excellence Partnerships (Project: HARSH, code: Y$\Pi$ 3TA-0560901), implemented through the National Recovery and Resilience Plan Greece 2.0 and funded by the European Union -- NextGenerationEU.
The authors thank the anonymous reviewers for their thoughtful feedback.
\end{acks}


\bibliographystyle{ACM-Reference-Format}
\bibliography{references}

\appendix
\clearpage
\onecolumn

\input{sections/aedescp}

\input{sections/fullrecandcorrectness}

\input{sections/appx_exp}



\end{document}

%% file: macro.tex
\newcommand{\punt}[1]{}
\newcommand{\cmnt}[1]{}

\definecolor{xxxcolor}{rgb}{0.8,0,0}
\newcommand{\XXX}[1]{{\color{xxxcolor} XXX: #1}\xspace}

\newcommand{\func}[1]{\texttt{#1}}
\newcommand{\var}[1]{\texttt{#1}}

\newcommand{\nosplit}{\linebreak}

\def\nohyphens{\hyphenpenalty=10000\exhyphenpenalty=10000}

\newcommand{\tilda}{\symbol{126}}

\newcommand{\ang}[1]{\langle #1 \rangle}
\newcommand{\Ang}[1]{\Big\langle #1 \Big\rangle}
\newcommand{\ceil}[1]{\lceil #1 \rceil}
\newcommand{\floor}[1]{\lfloor #1 \rfloor}

\newtheorem*{lemmanonum}{Lemma}
\newtheorem*{corollarynonum}{Corollary}
\newtheorem*{theoremnonum}{Theorem}
\newtheorem*{observationnonum}{Observation}
\newtheorem{property}[theorem]{Property}
\newtheorem{guarantee}{Guarantee}
\newtheorem{requirement}[theorem]{Requirement}
\newcounter{history}
\newcommand{\hist}[1]{\refstepcounter{history} {#1}}

\newcommand{\trevor}[1]{\textbf{[[[Trevor: #1]]]}}
\newcommand{\ajay}[1]{\textcolor{red}{#1}}
\newcommand{\y}[1]{\textcolor{green}{#1}}
\newcommand{\pf}[1]{\textbf{[[[PF: #1]]]}}

\newcommand{\True}{\mbox{\texttt{true}}}
\newcommand{\False}{\mbox{\texttt{false}}}

\newcommand{\AggStack}{{\sc SEC}}
\newcommand{\DCE}{{\sc SEC}}
\newcommand{\DEC}{{\sc SEC}}
\newcommand{\SEC}{{\sc SEC}}
\newcommand{\DCEPlus}{{\sc DCE+}}
\newcommand{\TRB}{{\sc TRB}}
\newcommand{\EB}{{\sc EB}}
\newcommand{\FC}{{\sc FC}}
\newcommand{\CCElim}{{\sc CC}}
\newcommand{\TSI}{{\sc TSI}}
\newcommand{\LCRQ}{{\sc LCRQ}}


\newtheorem{acknowledgement}[theorem]{Acknowledgement}
\newtheorem{observation}[theorem]{Observation}
\newtheorem{invariant}[theorem]{Invariant}


\newcommand{\chapref}[1]{Chapter~\ref{chap:#1}}
\newcommand{\secref}[1]{Section~\ref{sec:#1}}
\newcommand{\figref}[1]{Figure~\ref{fig:#1}}
\newcommand{\tabref}[1]{Table~\ref{tab:#1}}
\newcommand{\stref}[1]{step~\ref{step:#1}}
\newcommand{\thmref}[1]{Theorem~\ref{thm:#1}}
\newcommand{\lemref}[1]{Lemma~\ref{lem:#1}}
\newcommand{\insref}[1]{line~\ref{ins:#1}}
\newcommand{\corref}[1]{Corollary~\ref{cor:#1}}
\newcommand{\axmref}[1]{Proposition~\ref{axm:#1}}
\newcommand{\defref}[1]{Definition~\ref{def:#1}}
\newcommand{\eqnref}[1]{Eqn(\ref{eq:#1})}
\newcommand{\eqvref}[1]{Equivalence~(\ref{eqv:#1})}
\newcommand{\ineqref}[1]{Inequality~(\ref{ineq:#1})}
\newcommand{\exref}[1]{Example~\ref{ex:#1}}
\newcommand{\propref}[1]{Property~\ref{prop:#1}}
\newcommand{\obsref}[1]{Observation~\ref{obs:#1}}
\newcommand{\asmref}[1]{Assumption~\ref{asm:#1}}
\newcommand{\thref}[1]{Thread~\ref{th:#1}}
\newcommand{\trnref}[1]{Transaction~\ref{trn:#1}}
\newcommand{\lstref}[1]{listing~\ref{lst:#1}}
\newcommand{\Lstref}[1]{Listing~\ref{lst:#1}}

\newcommand{\subsecref}[1]{SubSection{\ref{subsec:#1}}}

\newcommand{\histref}[1]{\ref{hist:#1}}

\newcommand{\apnref}[1]{Appendix~\ref{apn:#1}}
\newcommand{\invref}[1]{Invariant~\ref{inv:#1}}

\newcommand{\Chapref}[1]{Chapter~\ref{chap:#1}}
\newcommand{\Secref}[1]{Section~\ref{sec:#1}}
\newcommand{\Figref}[1]{Figure~\ref{fig:#1}}
\newcommand{\Tabref}[1]{Table~\ref{tab:#1}}
\newcommand{\Stref}[1]{Step~\ref{step:#1}}
\newcommand{\Thmref}[1]{Theorem~\ref{thm:#1}}
\newcommand{\Lemref}[1]{Lemma~\ref{lem:#1}}
\newcommand{\Corref}[1]{Corollary~\ref{cor:#1}}
\newcommand{\Axmref}[1]{Proposition~\ref{axm:#1}}
\newcommand{\Defref}[1]{Definition~\ref{def:#1}}
\newcommand{\Eqref}[1]{eq(\ref{eq:#1})}
\newcommand{\Eqvref}[1]{Equivalence~(\ref{eqv:#1})}
\newcommand{\Ineqref}[1]{Inequality~(\ref{ineq:#1})}
\newcommand{\Exref}[1]{Example~\ref{ex:#1}}
\newcommand{\Propref}[1]{Property~\ref{prop:#1}}
\newcommand{\Obsref}[1]{Observation~\ref{obs:#1}}
\newcommand{\Asmref}[1]{Assumption~\ref{asm:#1}}
\newcommand{\reqref}[1]{Requirement~\ref{req:#1}}
\newcommand{\guarref}[1]{Guarantee~\ref{guar:#1}}

\newcommand{\Lineref}[1]{Line~\ref{lin:#1}}
\newcommand{\lineref}[1]{line~\ref{lin:#1}}
\newcommand{\algoref}[1]{Algorithm~\ref{algo:#1}}
\newcommand{\Algoref}[1]{{\sf Algorithm$_{\ref{algo:#1}}$}}

\newcommand{\Apnref}[1]{Section~\ref{apn:#1}}
\newcommand{\Invref}[1]{Invariant~\ref{inv:#1}}
\newcommand{\Confref}[1]{Conflict~\ref{conf:#1}}

\newcommand{\theqed}{$\Box$}
\newcommand{\nsqed}{\hspace*{\fill} \theqed}

\renewcommand{\thefootnote}{\alph{footnote}}
\newcommand{\ignore}[1]{}
\newcommand{\myparagraph}[1]{\noindent\textbf{#1}}

%

\newcommand{\lastup} {lastUpdt}
\newcommand{\lupdt}[2] {#2.lastUpdt(#1)}
\newcommand{\fkmth}[3] {#3.firstKeyMth(#1, #2)}

%
\newcommand{\nz}{\emph{restartable}\xspace}
\newcommand{\cas}[3] {CAS(#1, #2, #3)}

\newcommand{\qp}{\emph{quiescent-phase}\xspace}
\newcommand{\rdp}{\emph{$\Phi_{read}$}\xspace}
\newcommand{\wtp}{\emph{$\Phi_{write}$}\xspace}

\newcommand{\rd}{\emph{reader}\xspace}
\newcommand{\wt}{\emph{writer}\xspace}
\newcommand{\rl}{\emph{reclaimer}\xspace}
\newcommand{\rrc}{\emph{reader-reclaimer}\xspace}
\newcommand{\wrc}{\emph{writer-reclaimer}\xspace}

\newcommand{\knbr}{\emph{NBR}\xspace}
\newcommand{\ds}{\emph{data structure}\xspace}
\newcommand{\rb}{\emph{retireBag}\xspace}
\newcommand{\lb}{\emph{limboBag}\xspace}
\newcommand{\hw}{\emph{HiWatermark}\xspace}
\newcommand{\lw}{\emph{LoWatermark}\xspace}
\newcommand{\smr}{\emph{safe memory reclamation}\xspace}

\newcommand{\tid}[1]{$T_#1$}
\newcommand{\pred}{\emph{pred}\xspace}
\newcommand{\curr}{\emph{curr}\xspace}

\newcommand{\nbr}{NBR\xspace}
\newcommand{\nbrp}{NBR+\xspace}
\newcommand{\rcu}{RCU\xspace}
\newcommand{\qsbr}{QSBR\xspace}
\newcommand{\debra}{DEBRA\xspace}
\newcommand{\ibr}{IBR\xspace}
\newcommand{\geibr}{2GEIBR\xspace}
\newcommand{\cl}{crystallineL\xspace}
\newcommand{\cw}{crystallineW\xspace}
\newcommand{\hp}{HP\xspace}
\newcommand{\he}{HE\xspace}
\newcommand{\wfe}{WFE\xspace}
\newcommand{\ebr}{EBR\xspace}



\newcommand{\true}{\mbox{\textsf{TRUE}}}
\newcommand{\false}{\mbox{\textsf{FALSE}}}
\newcommand{\Insert}{\op{Insert}}
\newcommand{\DeleteMin}{\op{DeleteMin}}
\newcommand{\DeleteMax}{\op{DeleteMax}}
\newcommand{\Locate}{\op{Locate}}
\newcommand{\Restructure}{\op{Restructure}}
\newcommand{\CombinerDelete}{\op{CombinerDelete}}
\newcommand{\CCSynch}{\op{CCSynch}}
\newcommand{\HigherLevelInsertion}{\op{HigherLevelInsertion}}
\newcommand{\CAS}{\op{CAS}}
\newcommand{\FAO}{\op{FAO}}

\newcommand{\MIN}{\mbox{Min}}
\newcommand{\MAX}{\mbox{Max}}
\newcommand{\f}[1]{\mbox{\textit{#1}}} 
\newcommand{\x}[1]{\mbox{$\mathit{#1}$}} 

\newcommand{\here}[1]{{\bf [[[#1]]]}}  
\newcommand{\remove}[1]{}  
\newcommand{\er}[1]{{\bf [[#1--Eric]]}}
\newcommand{\myedit}[2]{{\color{#1}{#2}}\normalcolor}

\newcommand{\op}[1]{{\sf #1}} 


\newboolean{showcomments}
\setboolean{showcomments}{true} 

\newcommand{\ajcomment}[1]{%
  \ifthenelse{\boolean{showcomments}}%
    {\textcolor{red}{\small \textbf{[AjComment:} #1\textbf{\small ]}}}%
    {}%
}

\newcommand{\aj}[1]{%
  \ifthenelse{\boolean{showcomments}}%
    {\textcolor{blue}{\small \textbf{[Aj:} #1\textbf{\small ]}}}%
    {}%
}

\newcommand{\nikos}[1]{%
  \ifthenelse{\boolean{showcomments}}%
    {\textcolor{brown}{\small \textbf{[nikos:} [[#1]]\textbf{\small ]}}}%
    {}%
}

\newcommand{\nikaj}[1]{%
  \ifthenelse{\boolean{showcomments}}%
    {\textcolor{brown}{\small \textbf{[} [[#1]]\textbf{\small ]}}}%
    {}%
}

\newcommand{\ajedit}[1]{\textcolor{blue}{#1}} 

%% file: sections/introduction.tex
\section{Introduction}
\label{sec:intro}

Stacks are fundamental data structures utilized in various applications,
as well as in operating systems, and in system software. They support 
the \op{push}, \op{pop}, and \op{peek} operations for managing elements in a Last-In First-Out (LIFO) manner. Concurrent stacks are widely used as building blocks in concurrent pools~\cite{HS08}, shared freelists in garbage collection~\cite{yang2022deep, anderson2021concurrent}, and concurrent graph algorithms~\cite{nguyen2013lightweight}. 
Multiple threads may concurrently attempt to access the stack, which
often requires atomic access to a shared pointer \x{top} pointing to the topmost element of the stack,
thus resulting in heavy cache invalidation traffic. 
Therefore, naive concurrent stack implementations not only do not exhibit  speedup, 
but also they may incur drastic performance degradation. 
To mitigate this sequential bottleneck, several techniques have been proposed in the literature, 
including elimination~\cite[Section 11]{HS08}, software combining~\cite{OTY99,hendler2010flat,fatourou2012revisiting,FK17,KSW18}, 
timestamping~\cite{DHK15}, and combinations of these as elimination can be implemented on top 
of other techniques.

{\em Elimination} leverages the observation that two semantically opposite operations, 
such as a \op{push} and a \op{pop}, can effectively cancel each other out, 
leaving the data structure in a state as if neither occurred. 
This allows operations to complete without accessing the shared data structure, 
thereby reducing cache invalidation traffic. The Elimination-Backoff (EB) stack~\cite{hendler2004scalable}, 
for instance, employs a single elimination array which the threads use to exchange information in pairs
that allows them to discover whether they can eliminate the operations of each other. 
It utilizes an adaptive mechanism to determine the size of the elimination array at each point in time,
which results in good {\em elimination degree} (i.e., in a high number of eliminated operations on average), 
thus reducing contention on accessing the shared stack. 
However, its performance 
has been shown to be slower than that of stack implementations 
based on software combining~\cite{hendler2010flat} and the timestamp-based stacks~\cite{DHK15}.
\remove{\footnote{\textcolor{red}{I should somewhere say that DCE uses elimination proactively 
before accessing shared stack but EB accesses as a backoff. ]
In high concurrency scenarios, failing first at shared stack is potentially bad.}}}
The main overhead comes from the multiple \op{CAS} operations threads need to execute in order to discover whether elimination can occur.

\textit{Software combining}~\cite{OTY99,hendler2010flat,fatourou2012revisiting,FK17,KSW18} aims to reduce 
the synchronization overhead and the number of cache invalidations due to accesses to the shared \x{top} pointer by having a thread at each point in time, known as {\em combiner}, execute a batch of operations on behalf of other threads,
in addition to its own, after acquiring a global lock. Threads that do not act as combiners 
simply spin wait for the combiner to serve their operations and report their return values.
All threads announce their operations on a shared list
and the combiner traverses this list, 
applies the announced operations to the implemented data structure, reports return values, and releases the lock. 
This approach reduces the synchronization overhead at the \x{top} pointer (a contention hot spot), 
thereby cutting down the cache invalidation traffic. However, combining
sometimes unnecessarily reduces parallelism and results also in bottlenecks at high levels of concurrency, as we see in experiments.

The use of fetch\&add has been proved to be the state-of-the-art approach~\cite{FK11spaa,FK14,FKK22,LCRQ13, freudenthal1991process, hendler2010flat} 
for designing fundamental highly-efficient concurrent data structures, in particular FIFO queues~\cite{LCRQ13,roh2025aggregating}. 
For example, \LCRQ, published in PPoPP'13~\cite{LCRQ13}, 
which was recently shown to still be the fastest concurrent queue in~\cite{RK23}, 
is based on the simple idea of using fetch\&increment to implement two counters, 
one for assigning distinct sequence numbers to the \op{Enqueue} operations
and one to the \op{Dequeue} operations, so that the element inserted in the queue 
by the \op{Enqueue} with sequence number $i$ will be eliminated
by the \op{Dequeue} with the same sequence number. This is a natural idea
to follow for designing a FIFO queue as enqueuers insert elements on the one endpoint of the queue 
and dequeuers remove from the other. 

However, using fetch\&increment for designing a simple and efficient concurrent stack
entails a lot of challenges, 
since maintaining the LIFO property requires threads to operate on the same endpoint
to execute either a \op{push} or a \op{pop} operation. 
As a result, unlike queues, efficient use of fetch\&increment has not yet been realized for stack designs.


In this paper, we present a new stack implementation, called \DEC\ (Sharded Elimination and Combining), described in Section 3. \DEC\ efficiently uses fetch\&increment to achieve high performance by integrating elimination and combining in a unified design. These mechanisms share several components, enabling efficient implementation without duplicating costs. The algorithm thus merges elimination and combining seamlessly and introduces a novel design for a blocking linearizable stack.
Our experimental evaluation (\secref{experimentalevaluation}) studies the performance of \DCE\ stack under diverse settings, demonstrating that the new stack consistently outperforms state-of-the-art concurrent stacks. Moreover, the elimination and combining techniques are of independent interest and can be applied in other contexts, 
such as designing efficient concurrent deques or related data structures.

\SEC\ uses sharded elimination and combining 
to avoid the performance bottlenecks and limitations of prior elimination and combining techniques. 
It utilizes multiple aggregators, 
each assigned a subset of threads. 
Threads within each aggregator form batches of \op{push} and \op{pop} operations. 
Within each batch, threads cooperate 
to eliminate and combine operations.
Eventually, a batch is only left with either all \op{push} or all \op{pop} operations, 
which are then applied to the stack.


This batch-level elimination significantly reduces the number of threads 
that ultimately need to access the shared stack, thereby lessening contention at the stack's \x{top} pointer.
Combining within a batch reduces 
contention at the stack's \x{top} pointer
even further.
Our experiments as well as a recent aggregating funnels technique for implementing fetch\&add in software,
published in PPoPP'24~\cite{roh2025aggregating}, demonstrate the advantage of 
partitioning threads to multiple aggregators, akin to sharding
to reduce the overhead of a contention hot spot on large multicore machines.


In \DEC, having a distinct combiner for each batch where \op{push} operations are the majority,
results in increased parallelism, as it allows the combiners to  concurrently create substacks of nodes to append
in the shared stack. 
This also reduces the number of expensive 
synchronization primitives (such as CAS) performed by our algorithm, as each combiner attempts 
to append an entire substack in the shared stack by executing a singe CAS.
Similarly, the combiners of batches where \op{pop} operations are the majority,
attempt to update the \x{top} pointer once (with a single CAS) 
to remove the entire chain of topmost nodes from the stack in order to return them to the non-eliminated 
\op{pop} operations of their batches. 
Thus our approach allows for enhanced parallelism. 
Having a small number of combiners working concurrently differs from classic software combining approaches, 
such as CCSynch~\cite{fatourou2012revisiting} or flat-combining~\cite{hendler2010flat}, 
where a single combiner executes all operations on the stack sequentially.

Another novelty of \SEC\ is that it effectively uses two counters in each batch, 
implemented using fetch\&increment, to significantly simplify and accelerate the elimination and the combining
mechanisms.
A thread announces its operation in a batch
by simply incrementing one of these counters. These counters are later used to figure out
how many and which operations will be eliminated, and which operations will be applied 
to the shared stack. This way, in \SEC, elimination and combining are applied
at a significantly lower cost than in previous algorithms. 
For instance, 
the proposed algorithm requires only two fetch\&increment operations on two separate shared 
counters 
to support elimination, unlike the traditional \EB\ algorithm~\cite{hendler2004scalable}, 
which requires up to three \op{CAS} operations per \op{push} and \op{pop} pair. 
This results in a significant reduction in contention.
The use of counters also allows the algorithm to introduce a simple approach
for the discovery of the return values of non-eliminated \op{pop} operations. 

Through these innovations, our approach achieves better throughput
 without impacting the performance of operations disproportionately 
and demonstrates significant improvement over the existing state of the art algorithms 
on high thread count. 
\SEC\ is blocking because its elimination and combining mechanisms involve waiting. Our experiments suggest that \SEC\ does not have a clear advantage under low contention but it significantly outperforms all competitors at high thread counts.


\noindent
\textbf{Contributions} of this paper are: \\
$\bullet$ A lightweight mechanism for elimination based on thread sharding and the efficient use of
fetch\&increment. This mechanism is of independent interest and could be used to design 
other concurrent data structures (e.g. dequeues). \\ 
$\bullet$ A technique to tie up the elimination mechanism with a highly-efficient combining approach,
which results in enhanced parallelism when executing \op{push} and \op{pop} operations
on the shared stack. These mechanisms  result 
in significantly less contention when accessing the shared stack.\\
$\bullet$ A novel concurrent stack implementation in C++, which incorporates
the proposed elimination and combining techniques in a way that 
they have several components in common to avoid paying certain costs twice.
 \\
$\bullet$ A comprehensive experimental analysis to illustrate that the proposed stack implementation
outperforms all previous concurrent stacks in many cases on large scale NUMA systems.

\remove{

\textbf{Outline.}
The rest of the paper is organized as follows. In \secref{relatedwork}, we discuss state of the art algorithms which aim to mitigate the serialization bottleneck in the stacks. In \secref{algo}, we discuss in detail our algorithm and its code. This is followed by performance evaluation in \secref{experimentalevaluation} and conclusion in \secref{conclusion}.

}

\remove{
Threads form a batch through a \textit{freezing} mechanism. Specifically, every thread attempts to announce its operation within a batch.
\textcolor{blue}{Among them, the threads that announced their push and pop operations first attempt to decide the number of push and pop operations that are going to be the part of the batch. This event is called \textit{freezing} of a batch and the the thread that performs this is called a \textit{freezer} of the batch.
Until the \textit{freezer} decides the number of push and pop operations, all other threads that announced their operation in the batch simply wait for the batch to \textit{freeze}.}

Once a batch is frozen, threads eliminate the concurrently announced push and pop operations, using a simple arithmetic described later\aj{I should use more descriptive phrase than mere later}, leaving the batch with either only push or only pop operations.
The thread that announced the first non eliminated operation of the batch then becomes the \textit{combiner} and applies the batch's operation to the shared stack.
}

%% file: sections/relatedwork.tex
\section{Related Work}
\label{sec:relatedwork}

Several techniques in the literature~\cite{dodds2015scalable, hendler2004scalable, hendler2010flat, PZ18, SZ00}  
try to reduce the overhead due to the contention bottleneck when accessing the top pointer 
in stacks.
These leverage either the idea of exponential backoff, elimination, combining, timestamping, 
or a mix of these ideas to reduce the high overhead of synchronizing operations at the top of a stack.

Exponential backoff~\cite[Section 7.4]{HS08} is a well-known technique widely used in several settings. In this approach, a thread that observes contention backs off for a randomly chosen period,
giving competing threads a chance to finish their operations.
The backoff time is adjusted based on the observed contention. 
\DCE\ benefits from a simple backoff scheme which however serves a different goal. 
In \DEC, specific threads, called {\em freezers}, backoff for a small amount of time to increase the possibility of having more operations assigned to each batch, which might result in higher elimination and combining degrees.

Elimination has been proposed in \cite{hendler2004scalable} as an alternative backoff scheme for stacks. 
\EB\ \cite{hendler2004scalable} is a lock-free stack  which implements elimination as backoff
using a single elimination array. \EB's elimination mechanism is pretty heavy. Threads have to execute
three \op{CAS} operations to cooperate for eliminating a pair of operations.
Instead of using elimination as a backoff, \DEC\ proactively attempts to eliminate operations 
before accessing the shared stack. Its elimination mechanism is lightweight and requires only two fetch\&increment
operations to support the elimination of a pair of operations.

Software combining~\cite{OTY99,hendler2010flat,fatourou2012revisiting,FK17,KSW18} is a synchronization technique where a single combiner thread acquires a global lock 
and performs the requests of multiple other threads on their behalf, significantly reducing synchronization overhead 
and cache invalidations~\cite{hendler2010flat}. Software combining has been used to implement shared stacks~\cite{hendler2010flat,fatourou2012revisiting,FK17,KSW18}.
\DCE\ improves upon these stacks
by using multiple aggregators and batches. This results in enhanced parallelism and reduced
contention.

\textit{Timestamping} is utilized in~\cite{dodds2015scalable} to address the issue of a single hot point 
in stacks by exploiting the fact that linearizability allows concurrent operations to be reordered. 
Each element is associated with a timestamp during the execution of the corresponding \op{push} operation. 
This allows threads to push an element without requiring any synchronization on a single \x{top} pointer. 
However, it makes the pop and peek operations costlier, as the pop and peek operations 
are required to scan the timestamped elements and return an element with the highest timestamp from the stack for correctness.
In some of our experiments , threads perform \op{peek} operations to study this shift in overhead. Indeed, TSI~\cite{dodds2015scalable} shows worse performance than \DEC\ on read-heavy workloads.

\SEC\ borrows the idea of dispersing contention through nested partitioning from the aggregating funnels approach in~\cite{roh2025aggregating}, which implements a software Fetch\&Add primitive.
Like aggregating funnels, \DCE\ employs multiple aggregators and batches within each aggregator, but extends the design to implement a stack by introducing a) a novel batch-freezing mechanism, b) a new lightweight elimination scheme, c) a combining strategy supporting two types of operations (\op{push} and \op{pop}), and d) an efficient mechanism to return values to non-eliminated \op{pop} operations.
Thus, it substantially extends the techniques of~\cite{roh2025aggregating} and  applies them to design a concurrent stack that increases parallelism, reduces contention and scales well on large multicore machines.

%% file: sections/DCEalgorithm.tex
\section{Algorithm}
\label{sec:algo}




We start with a brief summary of how the algorithm works.
In \AggStack\ 
threads coordinate
in groups to reduce contention and to increase locality. 
This is achieved by utilizing multiple aggregators. Every thread is assigned to a fixed aggregator structure. 
Threads coordinate in the aggregator they have been assigned to. 
The way aggregators are used in \AggStack\ provides efficient implementations
for both an elimination scheme and a combining mechanism. 

The first goal of a thread executing an operation \op{op} is to discover
if \op{op} can be eliminated. To accomplish this, the operations invoked by the threads of an aggregator are split
into batches. Elimination occurs among the operations of the same batch. This 
requires a {\em freezing} mechanism to determine the operations that belong to each batch, as well as an array
within a batch where a \op{push} and a \op{pop} operation may meet and exchange information. 
If an exchange indeed occurs, the pair of \op{push} and \op{pop} operations is eliminated. 

The operations of a batch that are not eliminated, which are all of the same type
(either all \op{push} or all \op{pop} operations), 
are applied on a shared stack that \AggStack\ employs. To reduce contention, at most one {\em combiner} thread
may exist in each batch at each point in time. The role of the \textit{combiner} is to execute the non-eliminated operations 
of the batch on behalf of the threads that have initiated them, which perform
local spinning waiting for the \textit{combiner} to signal them that their operations have been served.
To parallelize the execution of \op{push}
operations among batches and reduce the number of expensive synchronization primitives (e.g., CAS)
performed on the state of the shared stack, the \textit{combiner} of each batch creates a local substack with 
the elements it wants to \op{push} in the shared stack. Then, it attempts to append the local
substack to the shared stack by executing CAS. 
Similarly, a \textit{combiner} that wants to pop a number of operations from the stack
attempts to pop them all together by atomically changing the top pointer of the stack 
to point to the appropriate subsequent stack node.

Section~\ref{sec:agg-descr} provides a detailed overview of the algorithm.
Section~\ref{sec:agg-details} presents the details of \AggStack\ pseudocode.

\subsection{Description of the Algorithm}
\label{sec:agg-descr}

\noindent
$\bullet$ {\bf Aggregators and batches.}
The algorithm uses $K$ {\em aggregators} and the shared stack.
Each thread is assigned to a particular aggregator. 
We denote by $P$ the maximum number of threads assigned to each aggregator.
Threads within an aggregator compete with one another.
The operations initiated by the threads assigned to an aggregator $A$ are split into batches. 
This way an aggregator employs more than one batch (possibly an infinite number of batches, if the execution is infinite). 
A {\em batch} is an object that enables threads to synchronize and cooperate in order to perform elimination and combining.
Each aggregator $A$ is implemented as a struct that contains a single pointer to its currently active batch. 
\\
\noindent
$\bullet$ {\bf Elimination and freezing.}
Among the operations that belong to a batch $B$ of an aggregator $A$,
if the number of \op{push} (\op{pop}) operations is fewer 
than the number of \op{pop} (\op{push}) operations, all the \op{push} (\op{pop}) operations of $B$ will be eliminated by \op{pop} (\op{push}) operations. 
Eventually, the batch will be left with only one type of operations that will be applied to the shared stack.
To implement elimination, the threads assigned to $A$ need a mechanism to decide the set $S$ of active operations that belong to 
$A$'s currently active batch $B$.
After $S$ has been decided, the active threads eliminate the biggest possible number of operations in $S$.

The threads decide which operations belong to $B$ through a {\em freezing} mechanism.
Specifically, $B$ stores two counters, namely \textit{pushCount} and \textit{popCount}.
To execute a \op{push} (or \op{pop}) operation, a thread announces the operation in $B$ by performing fetch\&increment on \textit{pushCount} (or \textit{popCount}).
Thus, \textit{pushCount} and \textit{popCount} represent the number of \op{push} and \op{pop} operations, respectively, that threads assigned to $A$ {\em announce} while $B$ is $A$'s active batch.
The value a thread sees in the counter it accesses is the {\em sequence number} of its active operation. 
The same sequence number can only be assigned to two operations of opposite type.

The threads that announce the first \op{push} and the first \op{pop} operations in  $B$ compete to decide which of them will become the {\em freezer} thread $f_B$ of $B$.
Thread $f_B$ will undertake the task of freezing $B$, i.e., it will indicate a point in time 
after which any subsequently announced operations will not belong to $B$.
A thread with sequence number $i$ that announced its operation into $B$ records the value it intends to \op{push} in the $i$-th element of an array, called \textit{eliminationArray} (stored in $B$).
It then spins, waiting for $f_B$ to freeze $B$.
A thread records the value to be pushed in \textit{eliminationArray} just after it obtains its sequence number.
This has the following advantages.
First, it prevents a thread performing a \op{pop} operation from waiting for the value to become available for exchange.
Second, it prevents the batch combiner from having to wait for the value.

\remove{
\st{A thread with sequence number $i$ that announced its operation into $B$ in the meantime, 
records its identifier into the $i$-th element of an array, called \textit{eliminationArray} (which is stored into $B$)
in an effort to speed up the process of its potential elimination after freezing, and then performs spinning
waiting for $f_B$ to freeze $B$.}
}

$B$ stores into counters \textit{pushCountAtFreeze} and \textit{popCountAtFreeze} 
the actual number of \op{push} and \op{pop} operations that belong to the batch,
i.e., the number of \op{push} and \op{pop} operations that have been recorded in \textit{pushCount} and \textit{popCount} by the time of the freeze.
Additionally, $B$ uses a boolean, called \textit{isFreezerDecided}, to choose $f_B$. 
\textit{isFreezerDecided} is initially FALSE.
The threads competing to become the freezer of $B$ attempt to set \textit{isFreezerDecided} to TRUE using Test\&Set. 
The winning thread will play the role of $f_B$
storing \textit{pushCount} and \textit{popCount} into \textit{pushCountAtFreeze} and \textit{popCountAtFreeze}.

Then, $f_B$ replaces the current working batch of the aggregator $A$ with a new one by changing $A$'s batch pointer to point to a new batch $B'$.
Any newly arriving thread will be assigned to $B'$ or to a subsequent batch.
With this pointer change, $f_B$ also signals the rest of the threads announced into $B$ to stop spinning.
Among them, those that have been announced after freezing and therefore  do not belong to $B$,
will attempt to be assigned to $B'$ or to a subsequent batch.
Each of the rest of the threads will attempt the elimination of their operation 
through \textit{eliminationArray}. 
Consider any such thread and let \op{op} be its operation which has sequence number $i$.
If there is another operation belonging to $B$ with sequence number $i$, 
then it should be of the opposite type for the elimination to occur. 
Otherwise, \op{op} has to be applied on the shared stack.
Threads whose operations are eliminated may immediately start the execution of new operations.
\\
\noindent 
$\bullet$ {\bf Combining.} Among the threads of a batch $B$ whose operations are not eliminated,
one plays the role of the combiner of $B$ undertaking the task of applying
all the non-eliminated operations of $B$
to the shared stack. This way a small number of threads 
may access the shared stack concurrently, resulting in reduced contention and enhanced performance. 

$B$'s combiner thread  is chosen to be the thread executing the first non-eliminated operation of $B$. Each batch
has a single combiner. To parallelize the execution of \op{push} operations by combiners of different batches,
each combiner creates a (local) substack containing all the non-eliminated \op{push} operations
in its batch (let them be $m$). Then, the combiner attempts to append the whole substack atomically to the shared stack by performing 
\op{CAS} on the top pointer of the stack. If the \op{CAS} is successful, all $m$ operations
are realized with a single \op{CAS}. This results in improved performance. 

Similarly, a combiner that has to perform $m$ \op{pop} operations will attempt to remove $m$ stack nodes atomically using \op{CAS}
to change the top pointer to point to the $m$-th node after it.
All threads whose operations are not eliminated and do not act as combiners perform local spinning waiting for the combiner 
of their batch to inform them that their operations have been applied on the shared stack. 
The combiner uses a boolean variable, called \textit{isBatchApplied}, 
stored into $B$, to signal the waiting threads to stop spinning and returns a substack to them to facilitate the discovery of their response values.
Specifically, a thread $q$ whose operation has sequence number $i$ reads the $i$-th node of the substack and returns its value if such a node exists else $q$ returns \texttt{EMPTY}.

To enhance performance, the freezer thread $f_B$ executes a short backoff before freezing $B$ (before \lineref{backoffspot})
to increase the {\em elimination degree} of \SEC, i.e., the average number of operations
that are eliminated, as well as its {\em combining degree}, i.e., the average number of operations that 
a combiner will serve. Experiments showed that this results in enhanced performance.

\subsection{Pseudocode and Details}
\label{sec:agg-details}

\begin{figure}
    \begin{algorithmic}[1]
    \LeftComment{Shared stack related variables}
    \State \texttt{\textbf{Struct} Node} \label{lin:typeNode}
    \Statex \quad \texttt{int} \f{value}  
    \Statex \quad \texttt{Node*} \f{next}
    \State \texttt{Node* stackTop} \label{lin:typeMainTop} 
    \State \texttt{int tid} \Comment{thread local id} \label{lin:tid}

    \LeftComment{Freezing, Elimination, and Combining related variables}
    \State \texttt{\textbf{Struct} Batch} \label{lin:typeBatch}
    \Statex \quad \texttt{int} \texttt{pushCount, popCount}  
    \Statex \quad \texttt{int} \texttt{pushCountAtFreeze, popCountAtFreeze} 
    \Statex \quad \texttt{Node*} \texttt{eliminationArray [P]} 
    \State \quad \texttt{Node*} \texttt{subStackTop}
    \Statex \quad \texttt{bool} \texttt{isFreezerDecided} 
    \Statex \quad \texttt{bool} \texttt{isBatchApplied} 
    \State \texttt{\textbf{Struct} Aggregator} \label{lin:typeAggregator}
    \Statex \quad \texttt{Batch* batch} 
    \State \texttt{Aggregator agg[K]} 
\end{algorithmic}
\caption{Key data structures and variables }
\label{fig:types}
\vspace{-0.5cm}
\end{figure}

\noindent
{\bf Data Structures and Variables.}
\Figref{types} shows the data structures and variables used in \AggStack.
The shared stack is implemented as a linked list of \texttt{Node}  (line~\ref{lin:typeNode}), 
each storing an integer \f{value} and a pointer to the next node. A shared \x{stackTop} pointer (line~\ref{lin:typeMainTop})
points to the stack's topmost node. 

The algorithm uses an array \x{agg} of $K$ \textit{Aggregator} objects (line~\ref{lin:typeAggregator}), 
each storing a pointer \f{batch} to its currently active batch. 
A batch (line~\ref{lin:typeBatch}) stores the four counters 
\f{pushCount}, \f{popCount}, \f{pushCountAtFreeze}, and \f{popCountAtFreeze},
as well as the boolean variables \f{isFreezerDecided} and \f{isBatchApplied}.
It also stores \f{eliminationArray}, the array that the operations of the batch use to eliminate one another. 

\begin{algorithm*}
    \caption{Push Algorithm.
    }
    \label{algo:push}
    \begin{algorithmic}[1]

        \Function{\texttt{Push}}{\texttt{int value}}: void \label{lin:orige-procread}
            \State \texttt{myAgg $\gets$ agg[tid/K]} \label{push:get-aggregator}
            \State \texttt{Node* myNode $\gets$ New Node(value, $\mathit{null}$)} \label{lin:allocatenode} \label{push:allocnode} 
            \While{\texttt{TRUE}} \label{lin:pushloopbegin}
                \State \texttt{myBatch $\gets$ myAgg.batch} \label{push:get-batch}
                \State \texttt{mySeqNum $\gets$ fetch\&increment(myBatch.pushCount)} \label{push:F&A}
                \State \texttt{myBatch.eliminationArray[mySeqNum] $\gets$ myNode}\label{push:write-value} 
                
                \If{\texttt{(mySeqNum} \texttt{==} \texttt{0} \texttt{\&}\texttt{\&} \texttt{!T\&S(myBatch.isFreezerDecided))}} \label{push:Freezer-Check} \Comment{freezing block}
                \State  \Call{FreezeBatch}{\texttt{myAgg, myBatch}} \label{push:FreezeBatch}
                \Else
                    \While{\texttt{myBatch == myAgg.batch}} \texttt{nop} \EndWhile \label{lin:pushNonFreezerWait} \Comment{non freezers wait for freezer to freeze the working batch}
                \EndIf \label{push:Freezer-Check-End}
                \medskip
                \If{\texttt{mySeqNum} \texttt{$<$} \texttt{myBatch.pushCountAtFreeze}} \label{lin:pushIncludedInBatch} \Comment{inclusion test.}
                    \If{\texttt{mySeqNum} \texttt{$\ge$} \texttt{myBatch.popCountAtFreeze}}  \label{push:elimination-check}\Comment{elimination test.}
                        \If{\texttt{mySeqNum} \texttt{==} \texttt{myBatch.popCountAtFreeze}}\label{lin:combiningBlockBegin} \Comment{combiner test}
                                \State \Call{PushToStack}{\texttt{myBatch, mySeqNum}} \label{push:push-substack-main}
                                \State \texttt{myBatch.isBatchApplied $\gets$ TRUE}  \label{push:release-all}
                        \Else 
                            \While{\texttt{!myBatch.isBatchApplied}} \texttt{nop} \label{lin:pushNonCombinerWait} \EndWhile 
                        \EndIf\label{lin:combiningBlockEnd}
                        \EndIf  \label{push:endInclusion}
                    \State \Return \texttt{TRUE} \label{lin:pushReturn}
                \EndIf \label{lin:pushNotIncludedInBatch}
        \EndWhile \label{lin:pushloopend}
    \EndFunction              
  
\medskip
    \Function{\texttt{FreezeBatch}}{\texttt{Aggregator myAgg, Batch* myBatch}} \label{lin:procfreeze}
        \State \texttt{myBatch.popCountAtFreeze $\gets$ myBatch.popCount} \label{lin:backoffspot} 
        \State \texttt{myBatch.pushCountAtFreeze $\gets$ myBatch.pushCount} 
        \State \texttt{myAgg.batch $\gets$ \Call{CreateNewBatch}{0,0,0,0,{0}, $\bot$, FALSE, FALSE}} \label{line:LPEliminatedOps} \label{lin:nonFreezerRelease} 
    \EndFunction   

\medskip

    \Function{\texttt{PushToStack}}{\texttt{Batch* myBatch, int MySeqNum}}
        \State \texttt{Node *top=$\bot$, *bot=$\bot$} 
        \State \texttt{int i $\gets$ MySeqNum} 
	\State {bot $\gets$ myBatch.eliminationArray[MySeqNum]} \label{lin:pushTopBotInit} 
        \While{\texttt{++i < batch.pushCountAtFreeze}}\label{push:addToSubStackBeg} \Comment{prepare a substack from \op{push} operations in myBatch}
            \While{\texttt{!myBatch.eliminationArray[i]}} \label{lin:waitPushNotAnn}
            \texttt{nop} 
            \EndWhile 
            \State \texttt{tempNode $\gets$ myBatch.eliminationArray[i]} \label{push:readToTemp}
            \State \texttt{tempNode.next $\gets$ top} \label{push:setTempNext}
            \State \texttt{top $\gets$ tempNode} \label{push:changeTopToTemp}
        \EndWhile \label{push:addToSubStackEnd}
        \While{\texttt{TRUE}} \label{push:addToSharedStackBeg} \Comment{Add the substack to the shared stack}
            \State \texttt{Node* tempTop $\gets$ stackTop}
            \State \texttt{bot.next $\gets$ tempTop} 
            \If{\texttt{CAS(stackTop, tempTop, top)}} \label{line:pushlp}
            \State \Return 
            \EndIf
        \EndWhile \label{push:addToSharedStackEnd}
    \EndFunction  
\algstore{myalgorithm}
\end{algorithmic}
\end{algorithm*}

\begin{algorithm*}
    \caption{Pop Algorithm}
    \label{algo:pop}
    \begin{algorithmic}[1]
    \algrestore{myalgorithm}

    \Function{\texttt{Pop}}{ }: int
        \State \texttt{myAgg $\gets$ agg[tid/K]} \label{pop:get-aggregator}
        \While{\texttt{TRUE}} \label{lin:poploopbegin}
            \State \texttt{myBatch $\gets$ myAgg.batch} \label{pop:get-batch}
            \State \texttt{mySeqNum $\gets$ fetch\&increment(myBatch.popCount)} \label{pop:F&A}
            \If{\texttt{(mySeqNum} \texttt{==} \texttt{0} \texttt{\&}\texttt{\&} \texttt{!T\&S(myBatch.isFreezerDecided))}} \Comment{freezing block} \label{pop:Freezer-Check}
            \State \Call{FreezeBatch}{\texttt{myAgg, myBatch}} \label{pop:FreezeBatch}
            \Else
                \While{\texttt{myBatch == myAgg.batch}} \texttt{nop} \EndWhile \label{pop:popNonFreezerWait}
            \EndIf \label{pop:Freezer-Check-End}
            \medskip
            \If{\texttt{mySeqNum} \texttt{$<$} \texttt{myBatch.popCountAtFreeze}} \Comment{inclusion test} \label{pop:popIncludedInBatch}
                \If{\texttt{mySeqNum} \texttt{$<$} \texttt{myBatch.pushCountAtFreeze}} \label{pop:Elimination-Check}\Comment{elimination test.}
                    \While{\texttt{!myBatch.eliminationArray[mySeqNum]}} \texttt{nop} \label{pop:Elimination-Wait} 
                    \EndWhile 
                    \State \Return \texttt{{myBatch.eliminationArray[mySeqNum].value}} \label{pop:Elimination-Return}
                \EndIf
                \medskip

                \If{\texttt{mySeqNum} \texttt{==} \texttt{myBatch.pushCountAtFreeze}} \label{pop:Combiner-Check}
                    \State \Call{PopFromStack}{\texttt{myBatch, mySeqNum}} \label{pop:PopFromSharedStack}
                    \State \texttt{myBatch.isBatchApplied $\gets$ TRUE}  \label{pop:Release-nonCombiners}
                \Else
                    \While{\texttt{!myBatch.isBatchApplied}} \texttt{nop} \label{pop:Wait-Combiner} 
                    \EndWhile
                \EndIf \label{pop:Combiner-Check-End}

            \State \texttt{\Return GetValue(myBatch,mySeqNum - myBatch.pushCountAtFreeze)}
            \EndIf \label{pop:endIncluded}
        \EndWhile \label{lin:poploopend}
    \EndFunction
\medskip
    \Function{\texttt{PopFromStack}}{Batch * \x{myBatch}, int \x{MySeqNum}}
        \State \texttt{int i $\gets$ MySeqNum} 
        \While{\texttt{TRUE}} \label{popAtomicRemoveLoopBeg}
            \State{\texttt{Node *top  $\gets$ stackTop}}
            \State{\texttt{Node *bot $\gets$ stackTop }}
	     \While{\texttt{++i < batch.popCountAtFreeze}}\label{pop:popFromStackBeg} 
            \If{\texttt{bot == null}} {\textit{break}} \EndIf
            \State{\texttt{bot $\gets$ bot.next}}
            \EndWhile
            \If{\texttt{CAS(stackTop, top, bot)}} \texttt{{break}} \label{line:poplp} 
            \EndIf
        \EndWhile \label{popAtomicRemoveLoopEnd}
        \State \texttt{myBatch.subStackTop $\gets$ top} 
    \EndFunction

\medskip

  \Function{\texttt{GetValue}}{\texttt{Batch *myBatch, int mySeqNum}}
        \State{\texttt{temp $\gets$ myBatch.subStackTop}}
        \For{\texttt{mySeqNum times}}
            \If{\texttt{temp == null}} {\Return \texttt{EMPTY}}
            \EndIf
            \State{\texttt{temp $\gets$ temp.next}}
        \EndFor
        \State \texttt{\Return temp.value}
\EndFunction

\end{algorithmic}
\end{algorithm*}

\noindent
{\bf Pseudocode for Push.}
A thread $q$ executing a \op{push} operation (\algoref{push}) first fetches the aggregator
it is assigned to based on its \textit{tid}  at line \ref{push:get-aggregator}.
In \SEC, threads are evenly distributed across aggregators (but more sophisticated schemes are also possible~\cite{roh2025aggregating}). For example, with two aggregators and ten threads, the first aggregator serves the first five threads and the second  the remaining five.
At line~\ref{lin:allocatenode}, $q$ allocates a new node with a \x{value} to be pushed and its \f{next} field equal to null . 
Then, the thread repeatedly does the following (loop of lines~\ref{lin:pushloopbegin}-\ref{lin:pushloopend}). 
It fetches a pointer to the aggregator's currently active \f{batch} $B$ into its local variable \x{myBatch}. 
It then performs a fetch\&increment on $B$'s \f{pushCount} (line~\ref{push:F&A})  thereby announcing its operation, 
and stores the returned value, which serves as  $q$'s sequence number in \x{mySeqNum}; 
$q$ uses its sequence number to determine the slot  to access in $B$'s \f{eliminationArray}
(line~\ref{push:write-value}).

Lines~\ref{push:Freezer-Check}-\ref{push:Freezer-Check-End} correspond to the freezing mechanism. 
Thread $q$ checks if it has received the sequence number $0$ among the threads that increment $B$'s \textit{pushCount}
(first condition of the if statement of line~\ref{push:Freezer-Check}). 
If $q$'s sequence number is indeed $0$, $q$ will execute the Test\&Set of line \ref{push:Freezer-Check}
to avoid a potential race with the thread executing the first announced \op{pop} in $B$, which might also attempt to become the freezer.
If the Test\&Set succeeds (i.e., returns $0$), $q$ sets $B$’s \f{isFreezerDecided} flag to \texttt{TRUE} and undertakes the role of the freezer $f_B$.
If $q$ does not get $0$ as its sequence number, it 
performs spinning (line~\ref{lin:pushNonFreezerWait}) waiting for $f_B$ to complete
the freezing phase by changing $B$'s \textit{batch} pointer to point to a new batch.


The freezer thread $f_B$ executes the FreezeBatch function (lines \ref{lin:procfreeze}-\ref{lin:nonFreezerRelease}), 
which stores a copy of \textit{pushCount} and \textit{popCount}
into \textit{pushCountAtFreeze} and \textit{popCountAtFreeze}, respectively. 
Thread $f_B$ also allocates a new batch $B'$ 
and sets the aggregator's \textit{batch} pointer to point to $B'$ (line~\ref{lin:nonFreezerRelease}).
This informs the other threads that freezing is complete.

At line \ref{lin:pushIncludedInBatch}, $q$ checks 
if its operation \op{op} belongs to $B$ by comparing \op{op}'s sequence number
with the value $f_B$ recorded into $B$'s \textit{pushCountAtFreeze}.
If \op{op}'s sequence number is greater than or equal to \textit{pushCountAtFreeze}, 
this means that $q$ incremented \textit{pushCount} after $f_B$ read 
\textit{pushCount}, so it does not belong to $B$ and has to retry 
applying its operation using a later batch 
(while loop at line~\ref{lin:pushloopbegin}). 

A thread $q$ whose \op{push} operation \op{op} belongs to $B$ executes the body of the if statement at line~\ref{lin:pushIncludedInBatch}. 
At line \ref{push:elimination-check}, $q$ checks if $\mathit{op}$ can be eliminated. 
If \op{op}'s sequence number is smaller to the number of \op{pop} operations
that belong to $B$, then \op{op} can be eliminated, so $q$ simply returns TRUE
at line~\ref{lin:pushReturn}. 
Otherwise, \op{op} cannot be eliminated, so $q$ executes lines~\ref{push:push-substack-main}-\ref{lin:combiningBlockEnd}.
At line~\ref{lin:combiningBlockBegin}, $q$ checks whether it should become a combiner
by comparing \x{mySeqNum} with \f{popCountAtFreeze}. 
If there are more \op{push} operations belonging to $B$ than pop, then 
the thread for which \x{mySeqNum} equals \f{popCountAtFreeze} becomes the combiner,
i.e., the combiner is the thread executing the \op{push}  operation with the smallest sequence number 
among those that have not been eliminated.
The combiner executes the \op{PushToStack} routine (line \ref{push:push-substack-main}) to apply the non-eliminated push
operations on the shared stack. 
The rest of the threads spin on line \ref{lin:pushNonCombinerWait}
waiting for the combiner to signal them that their operations have been applied
by setting the \textit{isBatchApplied} flag to TRUE (line \ref{push:release-all}).

The \op{PushToStack} routine is executed only by one thread at a time, namely the combiner at that time. 
\op{PushToStack} takes as arguments a pointer to $B$ and the sequence number \x{mySeqNum} of the combiner's operation. By the way the combiner is chosen,
all \op{push} operations with sequence numbers between \x{mySeqNum} and \f{pushCountAtFreeze} must be applied on the shared stack. 
The combiner creates a substack containing the nodes created by each of these operations (lines~\ref{push:addToSubStackBeg}-\ref{push:addToSubStackEnd}).
The last node of the substack is then linked to the current top node of the shared stack (lines~\ref{push:setTempNext}-\ref{push:changeTopToTemp}). 
This is done by tracking the bottom-most node of the substack in variable \x{bot}
(line~\ref{lin:pushTopBotInit}), and its top-most node in variable \x{top} (line~\ref{push:changeTopToTemp}).
During the procedure of creating the substack, the combiner might have to wait (line~\ref{lin:waitPushNotAnn}) 
until the thread that has been assigned the sequence number $i$ records its node in \x{eliminationArray}. 
Once written, the node is read at line~\ref{push:readToTemp} and
linked to the top of the substack at line~\ref{push:setTempNext}. Then, \x{top} is updated 
to this newly linked node at line~\ref{push:changeTopToTemp}.

The combiner repeatedly executes lines~\ref{push:addToSharedStackBeg} to~\ref{push:addToSharedStackEnd} 
until it successfully updates \x{stackTop}. This involves reading the current value of \x{stackTop}, 
linking the bottom-most node of the substack to it, and attempting a \op{CAS} to set \x{stackTop}
to be a pointer to the substack's topmost node. If the \op{CAS} succeeds, the function returns; otherwise, it retries.

\noindent
{\bf Pseudocode for Pop.} The announcement and freezing parts for a \op{pop} operation are similar to that for a push.
Specifically, a thread $q$ executing  a \op{pop} operation fetches the aggregator $A$
it is assigned to (line~\ref{pop:get-aggregator}) and $A$'s currently active batch $B$ (line~\ref{pop:get-batch}).
Then, it performs a fetch\&increment on $B$'s \f{popCount} (line~\ref{push:F&A})  to announce its operation, 
and uses the return value, which is stored into \x{mySeqNum}, 
as its sequence number. 
Lines~\ref{pop:Freezer-Check}-\ref{pop:Freezer-Check-End} correspond to the freezing code,
which is the same as for \op{push}.

At line~\ref{pop:popIncludedInBatch}, $q$ checks if it belongs to $B$ 
by comparing its sequence number with the value of \x{popCountAtFreeze}. 
If $q$'s sequence number is greater than or equal to that value,
$q$ must retry in a later batch.
If $q$ belongs to $B$, it checks at line~\ref{pop:Elimination-Check} 
whether its operation \op{op} can be eliminated. 
If \op{op} can be eliminated, $q$ spins at line~\ref{pop:Elimination-Wait} 
until the \op{push} operation that will eliminate it 
writes its node into \x{eliminationArray}[\x{mySeqNum}]. 
Once the node is available, $q$ completes the execution of \op{op} returning the value recorded in the node.

If \op{op} cannot be eliminated, $q$  checks whether it should become a combiner (line~\ref{lin:combiningBlockBegin}).
If there are more \op{pop} operations belonging to $B$ than \op{push}, then 
the thread executing the \op{pop} operation with the smallest sequence number 
among those that have not been eliminated becomes the combiner.
The combiner executes the PopFromStack routine to apply the pop
operations that are not eliminated on the shared stack. 
The rest of the threads 
wait for the combiner to signal them that their operations have been applied.

The combiner applies all \op{pop} operations 
with sequence numbers between \x{mySeqNum} and \x{popCountAtFreeze} on the shared stack. 
In \op{PopFromStack}, the combiner first records the current value of \x{stackTop} into the local variable \x{top},
and then uses the local variable \x{bot} to traverse as many stack nodes as the number of non-eliminated \op{pop} operations,
provided that the stack contains enough nodes. 
The combiner then executes a \op{CAS} to update \x{stackTop} to point to the same node as \x{bot}. 
If the \op{CAS} succeeds, the function returns; otherwise, the combiner retries the steps above.

To find its response, a thread $q$ that executes a non-eliminated \op{pop} operation calls \op{GetValue}.
In \op{GetValue}, threads traverse the removed substack to discover their return values based on their
sequence numbers. The thread with sequence number $\x{SeqNum} + i$ (where \x{SeqNum} is 
the sequence number of the combiner)
receives the $i$-th node from the removed substack, if it exists; otherwise, it returns null.

{\bf Peek.} We omit the pseudocode for \op{peek} operation, as it is simply a read of \x{stackTop}, similar to the Treiber stack~\cite{Treiber86}. 

The discussion on correctness and reclamation is provided in the supplementary material.

%% file: sections/shortRecandCorrectness.tex
\section{Reclamation}
Here, we briefly describe how we can integrate a reclamation algorithm into our stack implementation.
In \DEC, we deploy Brown’s implementation of the epoch-based reclamation algorithm (DEBRA) to reclaim batches and stack nodes~\cite{brown2015reclaiming}. Other reclamation algorithms~\cite{anderson2021concurrent, singh2021nbr, sheffi2021vbr, singh2025publish, kim2024your, singh2023simple, nikolaev2024family} can be applied in the same way.
In \op{pop}, a shared stack node is retired to the reclamation algorithm when a non-eliminated waiting thread retrieves its node from the substack returned by a combiner.
A batch is retired either by a freezer after elimination, when the batch is empty because all its operations were eliminated, or by a combiner when all its operations have been applied to the shared stack.

\section{Correctness and Progress}
\label{sec:correctness}
Our \AggStack\ algorithm is a linearizable stack implementation. 
A detailed proof of correctness appear in the appendix.

Consider an operation \op{op} that belongs to a batch $B$ of an aggregator $A$.
Let $q$ be the thread that initiated \op{op}.

If \op{op} is eliminated, let $\op{op}'$ be the operation that eliminated \op{op}. 
Recall that \op{op} and \op{op}' are operations of opposite type with the same sequence number. Moreover, both \op{op} and \op{op}' must have read the same batch. 
Let \op{op} is a push and \op{op}' is a pop (the reverse case is symmetric). 
Let $t$ and $t'$ be the times at which \op{op} and \op{op}' apply their fetch\&increment on 
\f{B.pushCount} and \f{B.popCount}, respectively, to announce their operations. 
Operation \op{op} writes its values at the slot represented by \x{mySeqNum} at line \ref{push:write-value} and \op{op}' reads the value returned by \op{op} at line \ref{pop:Elimination-Return}. 
We linearize both $\op{op}$ and $\op{op}'$ at the latest of $t$ and $t'$, the moment in time when $\op{op}$ and $\op{op}'$ eliminate each other.

Assume now that \op{op} was not eliminated. Then, there is a combiner thread $c_B$ for $B$,
which applied \op{op} ($c_B$ might be $q$ or a thread executing some other operation that belongs to $B$).
If \op{op} is a \op{push} operation, we linearize it at the point that $c_B$ successfully executes the \op{CAS} of line 
~\ref{line:pushlp}
in \op{PushToStack}. Similarly, if \op{op} is a \op{pop} operation, 
we linearize it at the point that $c_B$ successfully executes the \op{CAS} of line 
~\ref{line:poplp}
in \op{PopFromStack}. If more than one operation is linearized at the same point, we break ties
using  sequence numbers; operations with smaller sequence numbers are linearized first.

\begin{property}
\AggStack\ stack is blocking
\end{property}
In \AggStack, threads wait for freezer and combiner threads to perform their respective operations. 
In addition, threads executing a \op{pop} may wait for a \op{push} during elimination (line~\ref{pop:Elimination-Wait}).
Hence, the algorithm is blocking.


%% file: sections/experiments.tex
\section{Experimental Evaluation}
\label{sec:experimentalevaluation}

\noindent
{\bf Experimental Setting.}
The experiments were conducted on an Intel's Emerald Rapids (\textbf{Emerald}) machine, as well as on an Ice Lake-SP (\textbf{IceLake}) machine.  
Emerald consists of 2 NUMA nodes with 12 cores each, supporting 2-way hyperthreading for a total of 56 hardware threads. 
IceLake has 4 NUMA nodes with 12 cores each, also with 2-way hyperthreading, for a total of 96 hardware threads.
\SEC\ is implemeted in C++.
Our benchmarks, compiled with \texttt{C++14} and \texttt{-O3} optimization, run on Ubuntu 24.04 
using the mimalloc allocator and an epoch based reclamation algorithm Debra~\cite{brown2015reclaiming} to reclaim in \SEC. 

We have also conducted experiments on a 192 thread Intel Sapphire Rapids machine. 
The plots for the IceLake and Sapphire Rapids appear in the Section~\ref{sec:icelakeextexp} and~\ref{sec:rapidsextexp} and follow similar trends across all machines.


\noindent
{\bf Tested Algorithms.}
We test \SEC\ and compare its performance with five well-known stack implementations, namely the Treiber's stack ({\bf \TRB})~\cite{Treiber86}, 
the elimination backoff stack (\EB)~\cite{hendler2004scalable},
the flat-combining based stack (\FC)~\cite{hendler2010flat}, 
the CC-Synch based stack (\CCElim)~\cite{fatourou2012revisiting},
and the interval variant of the timestamp-based stack (\textbf{\TSI}) from \cite{dodds2015scalable}
(which is the most well performed variant presented in \cite{dodds2015scalable}). 

The implementations of \FC\ and \CCElim\ were taken from the benchmark framework 
provided with \CCElim\ \cite{fatourou2012revisiting}, and the implementations of \EB\ and \TSI\ from the benchmark in \cite{dodds2015scalable}.
We used the default settings from the \TSI\ benchmark. We also experimented with other settings (e.g., other {\em delay values}),
but the performance of the algorithm did not turn out to be better in these settings than in the default setting. 
\\
\\
\noindent
{\bf Methodology.} 
Each experiment runs for 5 seconds on a stack initially prefilled with $1000$ nodes. 
Similarly, varying the prefill size had no significant impact on performance.
During execution, threads randomly perform \op{push}, \op{pop}, or \op{peek} operations, with values drawn uniformly from a specified range. 
We report {\em throughput} (millions of operations per second)
averaged over five runs
across varying thread counts for three workload types: Read-heavy (90\% \op{peek}, 5\% \op{push}, 5\% \op{pop}; \textit{i.e.,} 10\% updates), Mixed (50\% \op{peek}, 25\% \op{push}, 25\% \op{pop}; \textit{i.e.,} 50\% updates), and Update-heavy (50\% \op{push}, 50\% \op{pop}; \textit{i.e.,} 100\% updates). 
We experimented with and without thread pinning and noticed no significant performance differences, therefore, we report results without pinning. Adding random work between operations did not also make any difference in the performance trends we observed.

On Emerald and IceLake, the system is oversubscribed after 56 and 96 threads, respectively.
Unless stated otherwise, our experiments with \SEC\ use two aggregators, with threads evenly distributed across them, as this configuration yielded the best performance.
For \SEC, the variance in throughput was below 5\% across all experiments.
\\
\\
\noindent
{\bf Analysis.}
\Figref{exp1} shows \SEC's throughput across the aforementioned workloads against the other competing stacks.


On Emerald (\Figref{exp1m1}),  
\SEC\ achieves up to 1.8-2.5$\times$ higher throughput than \FC\ and \CCElim\ across all workloads.
\FC\ and \CCElim\ restrict 
parallelism by executing entire \op{push} and \op{pop} operations sequentially.
This reduces the synchronization cost incurred at the shared top pointer 
but quickly becomes a bottleneck. 
In contrast, the combiners in \SEC\ execute (parts of) the operations they serve
concurrently, serializing only when accessing the top pointer.
Thus,  they execute shorter critical sections. This results in enhanced parallelism, and thus in better performance.

\begin{figure*}
	\centering
	\begin{minipage}{\textwidth}
		\begin{subfigure}{\textwidth}
			\includegraphics[width=0.33\linewidth, keepaspectratio]{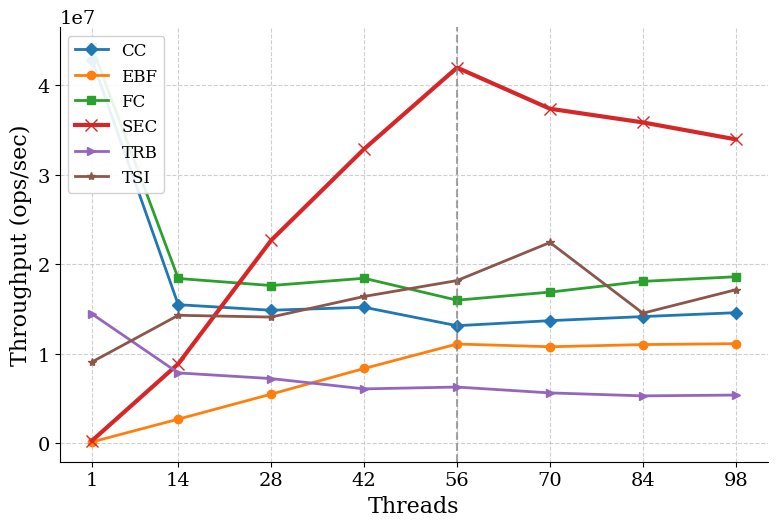}\hfill
			\includegraphics[width=0.33\linewidth, keepaspectratio]{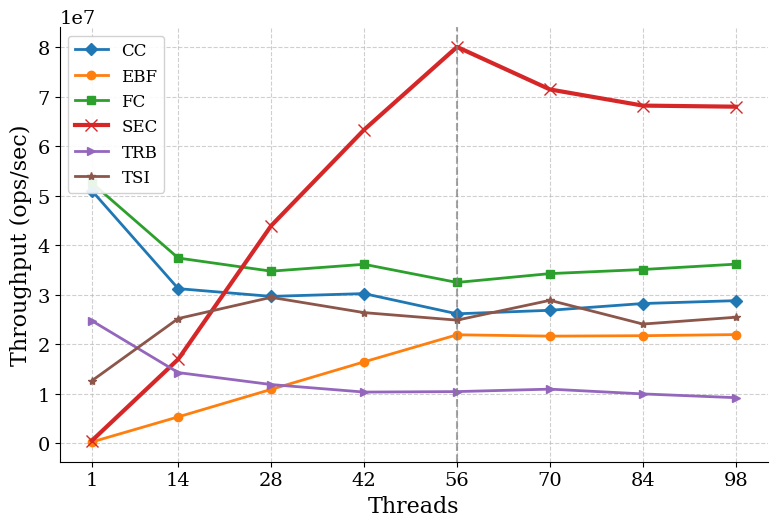}\hfill
			\includegraphics[width=0.33\linewidth, keepaspectratio]{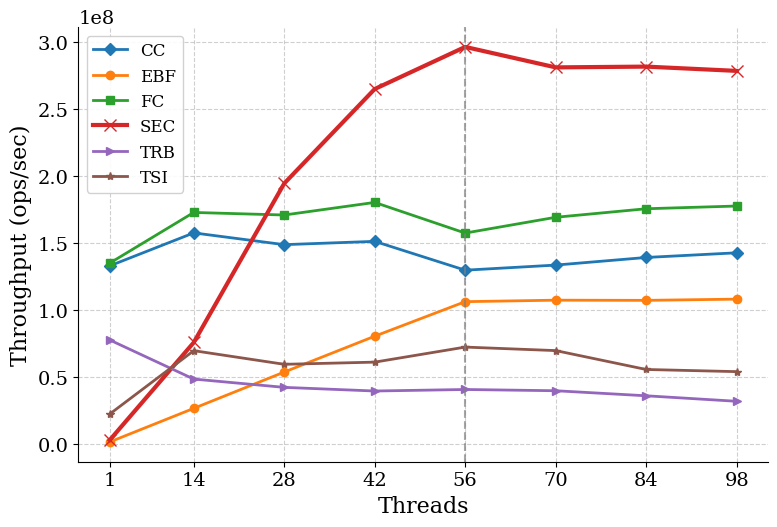}\hfill
			\caption{Throughput on Emerald. System is oversubscribed after 56 threads.}
			\label{fig:exp1m1}
		\end{subfigure}
		\begin{subfigure}{\textwidth}
			\includegraphics[width=0.33\linewidth, keepaspectratio]{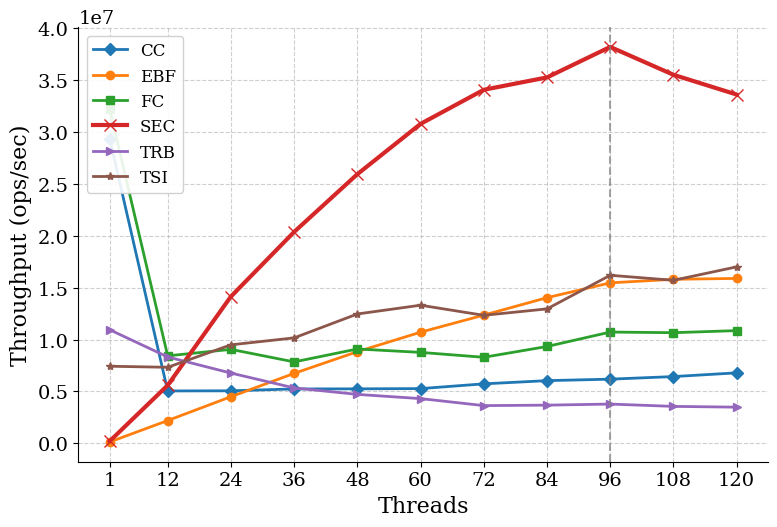}\hfill
			\includegraphics[width=0.33\linewidth, keepaspectratio]{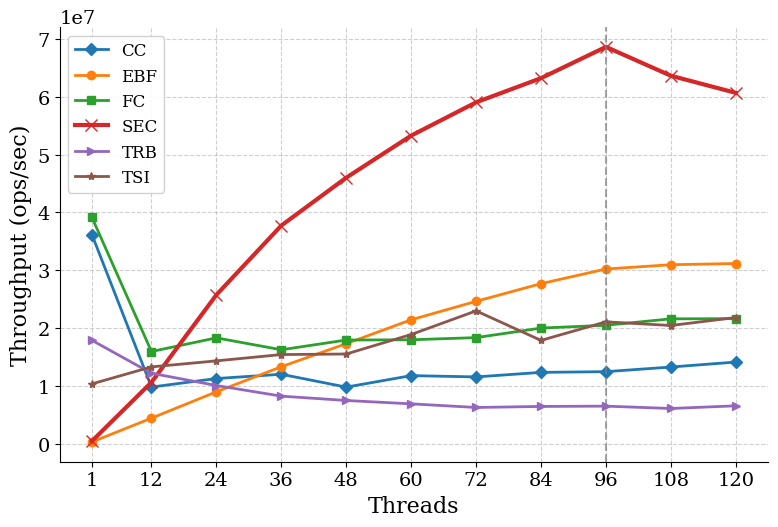}\hfill
			\includegraphics[width=0.33\linewidth, keepaspectratio]{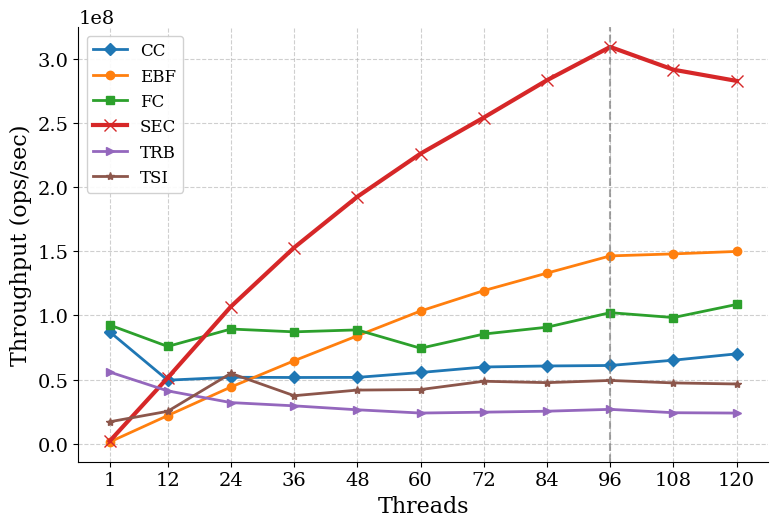}\hfill
			\caption{Throughput on IceLake. System is oversubscribed after 96 threads.}
			\label{fig:exp1m2}
		\end{subfigure}
	\end{minipage}
	\caption{Throughput. (Left) 100\% updates. (Middle) 50\% updates. (Right) 10\%updates. Y-axis: throughput in millions of operations per second. X-axis: \#threads. Number of aggregators used is two.}
	\label{fig:exp1}
\end{figure*}

\begin{figure*}
	\includegraphics[width=0.49\linewidth, keepaspectratio]{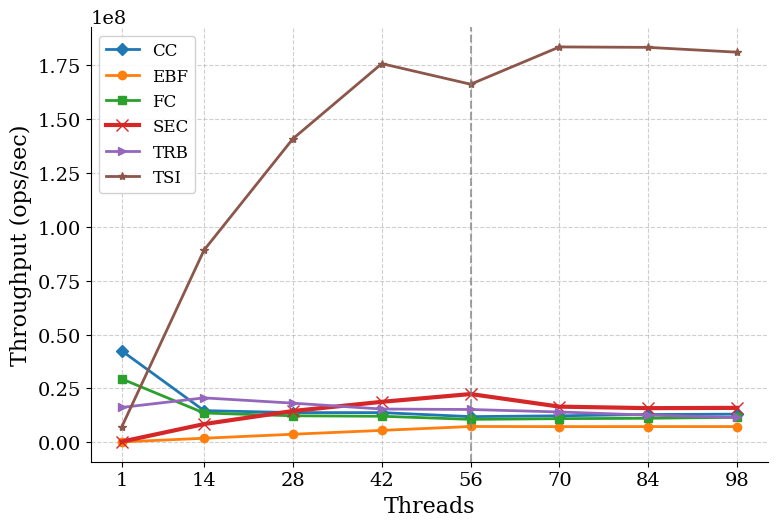}\hfill
	\includegraphics[width=0.49\linewidth, keepaspectratio]{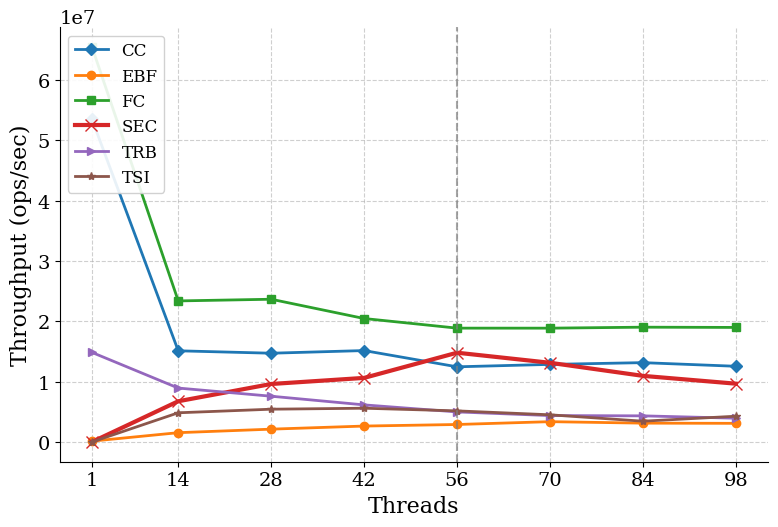}\hfill
	\caption{Throughput for push-only and pop-only workloads on Emerald. (Left) Push only. (Right) Pop only. Y-axis: throughput in millions of operations per second. X-axis: \#threads. Number of aggregators used is two.}
	\label{fig:exp2m1}
\end{figure*}

\begin{figure*}
	\centering
	\includegraphics[width=0.249\linewidth, keepaspectratio]{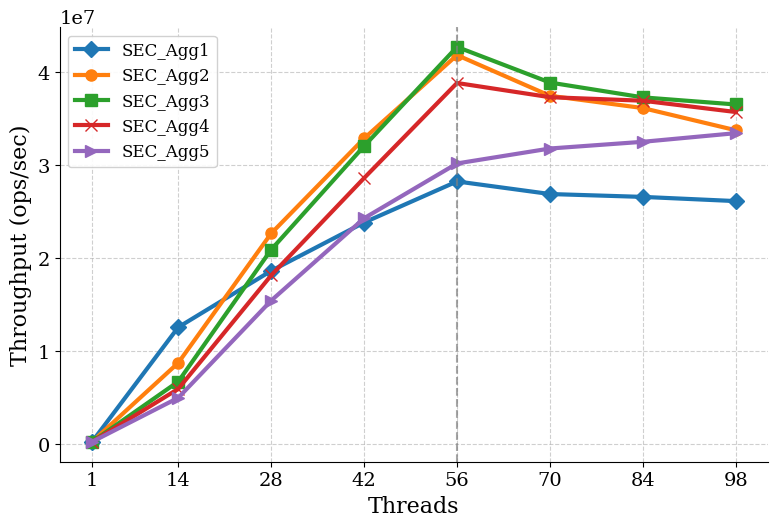}\hfill
	\includegraphics[width=0.249\linewidth, keepaspectratio]{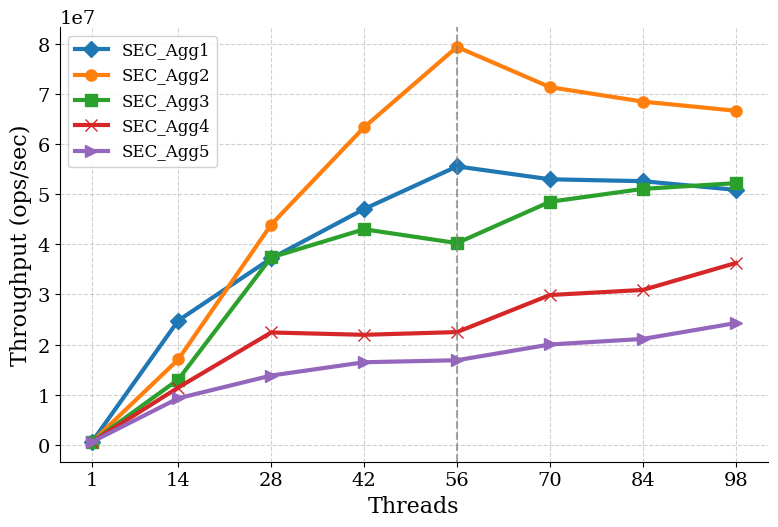}\hfill
	\includegraphics[width=0.249\linewidth, keepaspectratio]{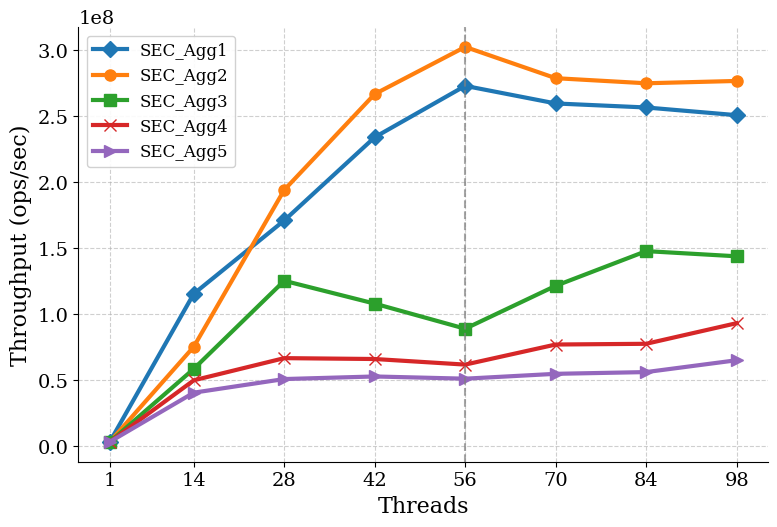}\hfill
	\includegraphics[width=0.249\linewidth, keepaspectratio]{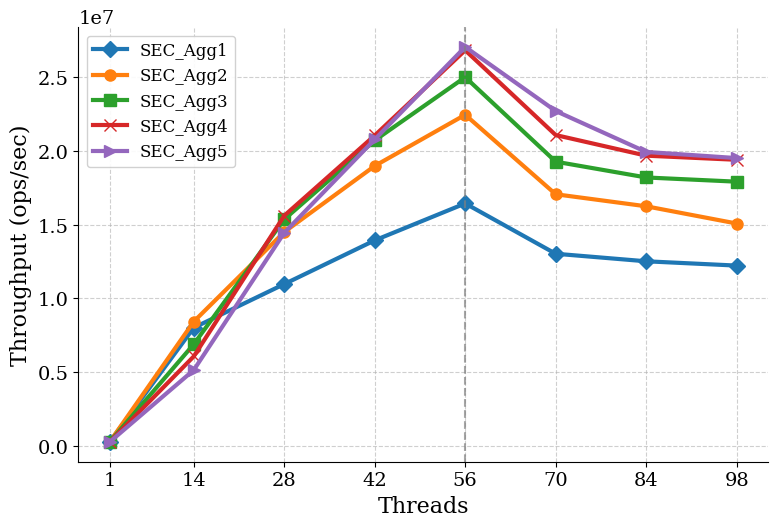}\hfill
	\caption{Comparing SEC throughput with various number of aggregators on Emerald. From left to right, 100\% updates, 50\% updates, 10\%updates, 100\%push-only. Y-axis: Throughput. X-axis: \#threads. SEC with 1 aggregator is labeled as SEC\_Agg1.}
	\label{fig:exp3m1}
\end{figure*}


In \Figref{exp1m1}, we see that 
\SEC\ outperforms \TSI, especially at large thread counts.
\TSI's performance in comparison to \SEC\ degrades as the update rate decreases (50\% and 10\%). 
This is due to the high overhead of \op{pop} and \op{peek} operations in \TSI.
To achieve synchronization-free push operations, 
\TSI\ shifts  overhead onto \op{pop} and \op{peek} operations. 
At 100\% update rate (50\% pushes and 50\% pops), the balanced mix of \op{push} and \op{pop} operations
allows the synchronization-free \op{push} operations in \TSI\ to offset the cost of \op{pop} operations, boosting performance. 
By contrast, at 50\% and 10\% update rates, the larger proportion of \op{pop} and \op{peek} operations 
increases overhead reducing performance.

On IceLake (\Figref{exp1m2}), \SEC\ is up to 2-2.3$\times$ faster than \TSI\ and other competitors,
scaling better across all workloads.
\TSI\ uses x86-specific RDTSCP instruction to define time intervals for efficient elimination.
To achieve this, it introduces delays between start and end of the time intervals that increases latency and negatively impacts \TSI's performance.

Among the other algorithms only \EB\ scales well. This is due to its elimination-based backoff mechanism.
However, it remains up to 2.6$\times$ slower than \SEC.
Its slower throughput stems from the fact that it employs a heavier elimination mechanism than \SEC. 
\EB\ requires three \op{CAS} to achieve elimination, whereas \SEC\ deploys a faster elimination mechanism 
based on fetch\&increment and employs sharding which reduces contention.  
Moreover, \SEC\ deploys its elimination mechanism proactively, whereas \EB\ uses elimination only as a fallback.
Also, the way the elimination array is used in \EB\ may result in pairs of operations
that can be eliminated but they are not, thus limiting the elimination degree of the approach. 
For instance, 
a \op{push} may wait on some slot of the elimination array without being eliminated although
there might be concurrent \op{pop} operations, which however have chosen other slots in the array.
On the contrary, in \SEC, the use of the \x{PushCounter} and \x{PopCounter}
ensures that the elimination degree is optimal within each batch.

\Figref{exp2m1} focuses on push-only and pop-only workloads on the Emerald machine.
Experiments for IceLake show similar performance trends and are provided in the Section ~\ref{sec:icelakeextexp}. 
These experiments exhibit the throughput of the tested algorithms in the absence of elimination 
and highlight techniques that impose asymmetric overhead on stack operations.
For example, \TSI\ has fast \op{push} operations but incurs high overhead for \op{pop} operations, whereas the other techniques, including \SEC, deliver similar performance for both \op{push} and \op{pop} operations.
Specifically, \TSI\ is up to 6$\times$ faster than \SEC\ in the push-only benchmark but \SEC\ is up to 3$\times$ faster than \TSI\ in the pop-only benchmark.
\TSI\ is faster in push-only workload because its push operations avoid synchronization at the shared top pointer by inserting  nodes into thread-local pools.
However, this design shifts the overhead to \op{pop} and \op{peek} operations, 
which becomes evident in the pop-only workload where \TSI\ incurs a significant slowdown due to the disproportionate overhead on \op{pop} operations.   

\Figref{exp3m1} shows the impact of the number of aggregators on the performance of \SEC\ where we vary the number of aggregators from 1 to 5. 
We denote a configuration with $m$ aggregators as `Agg$m$'. For example, \SEC\ with 1 aggregator is labeled \SEC\_Agg1.

In push-only workloads (rightmost plot in \Figref{exp3m1}), a higher number of aggregators is preferable (two or more),
as it better distributes the contention among threads: first through the aggregators and then through the batches within each aggregator.
As a result, the overhead of \textit{freezing} and \textit{combining} is spread across multiple groups of threads improving performance. Note that in this workload no elimination is possible.

For workloads with 100\% update rate (leftmost plot in \Figref{exp3m1}) two to four aggregators are preferable because in this setting threads are able to disperse contention between multiple aggregators and yet are able to take advantage of combining and elimination. 
On the other side, 
using five aggregators disperses contention between threads too much such that chances of threads taking advantage  elimination and combining reduces. For one aggregator, the opposite happens, where overhead of thread contention due to \textit{freezing} and \textit{combining} gets concentrated in a single aggregator.
For 50\% and 10\% update rate (middle plots in \Figref{exp3m1}) having one or two aggregator achieves the sweet spot for threads to be able to disperse the contention and yet be able to take advantage of combining and elimination.


We have chosen two aggregators for running all our workloads as it turns out to be the best setting in most cases.

\section{Conclusion}
\label{sec:conclusion}

We introduced SEC (Sharded Elimination and Combining), a linearizable concurrent stack that unifies elimination and combining through nested sharding. By partitioning threads into aggregators and then into batches, \SEC\ reduces contention at the shared stack and enables multiple combiners to execute in parallel, departing from traditional flat-combining approaches.
Key to efficiency in \SEC\ is its lightweight use of fetch\&increment counters, which simplify both elimination and combining and eliminate the need for costly synchronization in the common case. 
This incurs low overhead compared to prior algorithms such as EB, CCSynch, and flat-combining, while maintaining good performance.

Our evaluation shows that SEC consistently outperforms state-of-the-art stacks across diverse workloads and machines. Beyond stacks, the novel sharded elimination and efficient combining are of independent interest and can be applied to other concurrent data structures, such as deques. 



%% file: sections/aedescp.tex
\section{Artifact Description}
\label{sec:aedescp}
This section provides a step by step guide to run our artifact in a docker container. The artifact extends the popular Brown's setbench benchmark publicly available at: \url{https://gitlab.com/trbot86/setbench}.

The artifact can be found at the following links:
\begin{itemize}
    \item zenodo (most recent version v4 with docker image): \\
    \url{https://zenodo.org/records/18109078}.
    \item github (repo without docker image):\\
    \url{https://github.com/ConcurrentDistributedLab/combXAgg}
\end{itemize}

If you prefer to use the artifact directly without using the docker container please refer to the accompanying README file in the source code.


The following instructions will help you load and run the provided Docker image within the artifact downloaded from Zenodo link.
Once the docker container starts you can use the accompanying README file to compile and run the experiments in the benchmark.
\\
\\
\myparagraph{Steps to load and run the provided Docker image:}

Note: Sudo permission may be required to execute the following instructions. 

\begin{enumerate}
    \item 
    Install the latest version of Docker on your system. We tested the artifact with the Docker version 28.3.0, build 38b7060 on Ubuntu 24.04. Instructions to install Docker may be found at \url{https://docs.docker.com/engine/install/ubuntu/}. Or you may refer to the “Installing Docker” section at the end of this README.

    To check the version of docker on your machine use:
\begin{verbatim}
$ docker -v
\end{verbatim}

    \item Download the artifact from Zenodo at URL:\\
    \url{https://zenodo.org/records/18109078}.

    \item Extract the downloaded folder and move to \\
    \textit{sec\_setbench/} directory using $cd$ command.
    \item Find docker image named \textit{sec\_setbench.tar.gz}\\
    in sec\_setbench/ directory. And load the downloaded docker image with the following command:
\begin{verbatim}
 $ sudo docker load < sec_setbench.tar.gz
\end{verbatim}
\item 
Verify that image was loaded:
\begin{verbatim}
 $ sudo docker images
\end{verbatim}
\item Start a docker container from the loaded image:
\begin{verbatim}
 $ sudo docker run -it --rm sec_setbench
\end{verbatim}
\item Invoke $ls$ to see several files/folders of the artifact: Dockerfile, README.md, common, ds, install.sh, lib, microbench, sec\_experiments, tools.
\end{enumerate}

Now, to compile and run the experiments you could follow the instructions in the README file.

%% file: sections/fullrecandcorrectness.tex
\section{Correctness and Progress}
\label{sec:correctness}

We begin by stating a sequence of observations derived from the pseudocode, followed by a proof that \AggStack\ is linearizable.

Consider now an operation \op{op} that belongs to a batch $B$ of an aggregator $A$, and let $q$ be the thread that initiated \op{op}.
	
\begin{observation}
	Each batch of $A$ has exactly one freezer.
\end{observation}

This follows directly from the pseudocode.
At most two threads can compete to become the \textit{freezer}: one executing a \op{push} and one executing a \op{pop}. Both operations must have sequence number $0$, i.e., they are the first of their type in the batch. Only one can succeed in its \texttt{Test\&Set} (lines~\ref{push:Freezer-Check} and~\ref{pop:Freezer-Check}) and subsequently invoke \func{freezebatch}. Hence, each batch has exactly one freezer.

\begin{observation} 
A combiner applies either all push or all pop operations on the stack.
\end{observation}
This follows from the way threads announce their operations on a batch’s \textit{eliminationArray} from left to right.
Three cases can arise:
\begin{enumerate}
	\item A batch has more pushes than pops. In this case, after elimination, only \op{push} operations remain. The thread executing the first non-eliminated push becomes the combiner.
	\item A batch has more more pops than pushes. In this case, after elimination, only \op{pop} operations remain. Thhe thread executing the first non-eliminated pop becomes the combiner.
	\item A batch has equal number of pushes and pops. In this case, after elimination, the batch is empty and no operations remain to be executed. Precisely, in push, threads fail the check at line~\ref{push:elimination-check} and in pop, threads fail the check at line~\ref{pop:Elimination-Check}.
	Hence, combiner is not invoked.
\end{enumerate}

\begin{observation}
	Each batch of $A$ has exactly one combiner. 
\end{observation}
	Threads, perform freezing, elimination and combining in that order. Combining is required only if operations remain within a batch after elimination. These remaining operations must be either all \op{push}or all \op{pop}.
	
	Without loss of generality, consider a batch containing only \op{push} operations. The combining step must be performed by the thread executing the first non-eliminated \op{push} in the batch.
	Assume that the \x{pushCount} variable, incremented using $fetch\&increment$, does not wrap around.
	Within a batch, $fetch\&increment$ guarantees that no two threads obtain the same sequence number.
	During freezing, \x{pushCountAtFreeze} is assigned a unique sequence number.
	
	Consequently, during the combining test at line~\ref{lin:combiningBlockBegin}, only one thread's sequence number can match the batch's \x{popCountAtFreeze}. Hence, only one thread can enter the combining block at line~\ref{lin:combiningBlockBegin} in Algorithm~\ref{algo:push}.
	
	The argument is symmetric for a \op{pop} operation: only one thread's sequence number can match the batch's \x{pushCountAtFreeze}, and thus only one thread can enter the combining block at line~\ref{pop:Combiner-Check} in Algorithm~\ref{algo:pop}.

\begin{lemma}
	The linearization point of every non-eliminated operation lies within its execution interval.
\end{lemma}
\begin{proof}	
	Let \op{op} be a non-eliminated operation. Then there exists a combiner thread $c_B$ for batch $B$ that applies \op{op}. 
	(The combiner $c_B$ might be the thread $q$ that initiated the \op{op}, or another executing a \op{pop} operation in $B$.)
	
	If \op{op} is a \op{push}, we linearize it at $c_B$'s successful execution of the \op{CAS} at line~\ref{line:pushlp}
	in \op{PushToStack}. Similarly, if \op{op} is a \op{pop}, 
	we linearize it at $c_B$'s successful execution of the \op{CAS} at line~\ref{line:poplp}
	in \op{PopFromStack}. 
	Note that $c_B$ may push or pop several operations with a single \op{CAS}. All such operations are linearized in order of their sequence numbers, with the operation having smaller sequence number linearized first.
\end{proof}	

\begin{lemma}
	The linearization point of an eliminated \op{pop} operation occurs immediately  before the linearization point of its corresponding eliminated \op{push} operation, and both points lie within the execution interval of the \op{pop}.
	\label{lem:elimopslp}
\end{lemma}
\begin{proof}
	
Consider an eliminated pair of \op{push} and \op{pop} operation. Both operations are linearized at line~\ref{pop:Elimination-Return}, when the \op{pop} exchanges its value with the corresponding \op{push} during the \op{pop}'s execution interval. No other operation can be linearized between these two eliminated operations.
\end{proof}

\begin{invariant}
At all times $t$, \x{stackTop} represents the state of the stack that would result from sequentially executing all operations linearized before $t$, in order of their linearization points.
\label{inv:validstack}
\end{invariant}
 We prove the invariant by induction over a sequence of linearizing steps performed by operations on the shared stack. 
 
 \textbf{Induction Base:} Initially, the stack is empty, meaning no operations have modified it. If there were a sequence of eliminated operations they must cancelled each other out, so the net effect is that the stack remains empty. The invariant therefore holds vacuously.
 
 \textbf{Induction Hypothesis:}
 Assume the invariant holds up to a step $i$, due to some operation \op{op}.
 
 \textbf{Induction:} Consider the $(i+1)_(th)$ operation.
\begin{itemize}
\item
If it is an eliminated operations, it does not modify the stack and is cancelled out with the $i_{th}$ operation, leaving the stack unchanged. Moreover, there can be no other operation between the operations $i$ and $i+1$.

\item 
If it is a combiner executing a \op{CAS} to add non-eliminated \op{push} operations, the combiner sequentially adds the corresponding nodes to the top of the stack, in the order of the operations within the batch.

\item
If it is a combiner executing a CAS to perform $k$ \op{pop} operations, the combiner sequentially removes the top $k$ nodes from the stack, following the order of the pop operations in the batch.
\end{itemize}

 In all cases, the set of operations performed during step $i+1$ preserves the invariant.
 
 Therefore, by induction, the invariant holds for all steps.

\begin{lemma}
	The response of each operation is consistent with its linearization point.
\label{lem:resp}
\end{lemma}

Only \op{peek} and \op{pop} operations produce responses, so we will consider only these cases. 

\op{Peek} operations are linearized when \x{stackTop} is read. They response is the the value at the top node of the shared stack. 

Eliminated \op{push} or \op{pop} operations are linearized together such that the \op{pop} operation is immediately linearized after the matching \op{push}, with no other operation linearized between them. The \op{pop} returns the value exchanged with its matching \op{push}. In both the cases, the claim follows immediately.

Next, consider a batch of \op{pop} operations, $pop_1$, $pop_2$, $\cdots$, $pop_k$, linearized in order by a successful CAS on the \x{stackTop}. The thread executing $pop_1$ is the combiner that executes the CAS to remove top $k$ nodes from the shared stack. 
The response of $pop_i$ is the value of $i_{th}$ node removed from the stack, where $1$ $>$ $i$ $\leq$ $k$ (see the \func{GetValue()}). By Invariant~\ref{inv:validstack}, these responses respect the linearization order.

From Invariant~\ref{inv:validstack} and Lemma~\ref{lem:resp} the following theorem follows.
\begin{theorem}
	\AggStack\ is a linearizable stack implementation. 
\end{theorem}


%% file: sections/appx_exp.tex
\section{Extra Experiments on 56 Threads Machine: Emerald}

In \tabref{exp4m3}, we study the batching, the elimination and the combining degree of \SEC. 
The different columns show measurements for different update rates. 
We see that the batching degree and the elimination degree increase as the update rate increases, 
whereas the combining degree remains the same. 
For instance, in the 100\% update rate setting, the average size of a batch (across different thread counts) 
is 41, which is the largest among all settings, of which 78\% of the operations within a batch are eliminated.
Therefore, the main  performance advantages in \SEC\ come from efficient batching and elimination.


\begin{table}[ht]
\centering
\begin{tabular}{|c|c|c|c|}
\hline
      & \multicolumn{3}{c|}{\textbf{SEC}} \\ \hline
\textbf{Workload $\rightarrow$}      & \textbf{100\% upd} & \textbf{50\% upd} & \textbf{10\% upd} \\ \hline
Batching Degree    & 17.8    & 17.2   & 14  \\ \hline
\%Elimination      &  79\%   &  79\%  & 77\%  \\ \hline
\%Combining        & 21\%    &  21\%  & 23\%   \\ \hline
\end{tabular}
\caption{Batching degree (average size of batches during an execution), \%elimination (percent of operations eliminated per batch), and \%combining (percent of operations not eliminated per batch) in SEC for EXP1 on shuttle. Column label 100 implies workloads with 100\% updates and so on.}
\label{tab:exp4m3}
\end{table}

\newpage
\section{Experiments on 96 Threads Machine: IceLake}
\label{sec:icelakeextexp}

\begin{figure}[!htb]
\centering
    \includegraphics[width=0.33\linewidth, keepaspectratio]{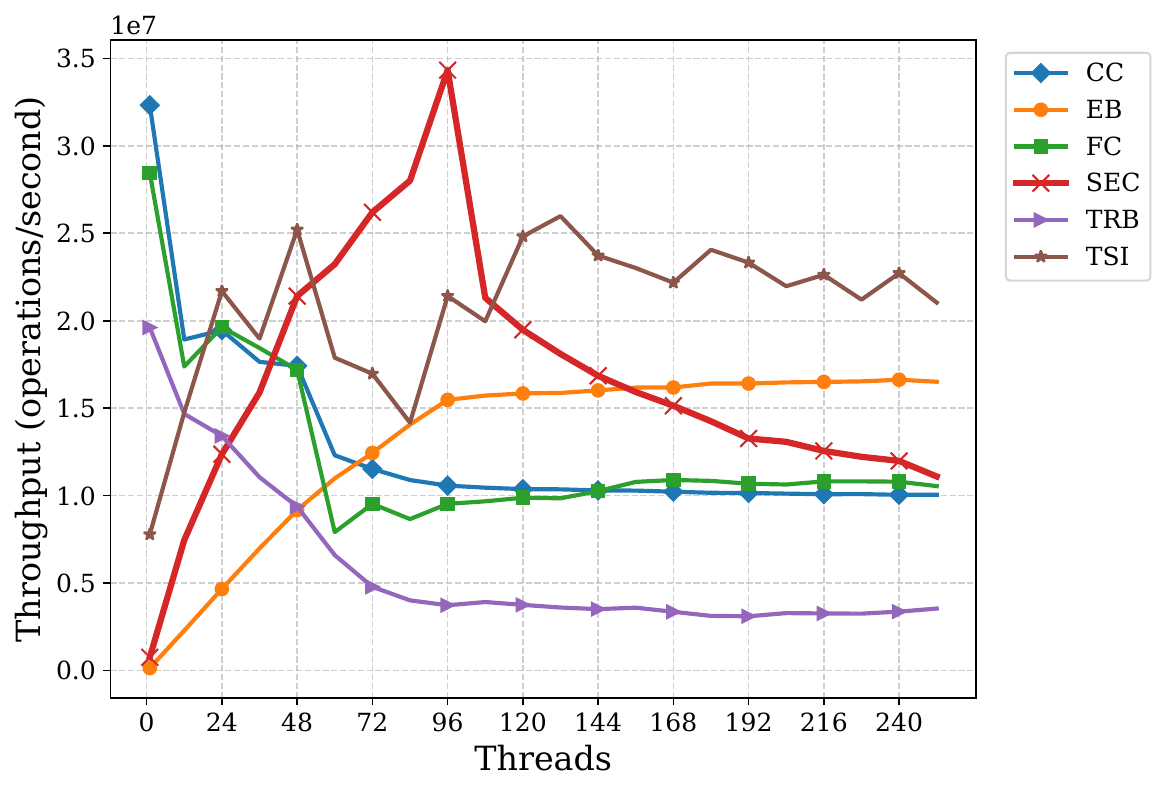}\hfill
    \includegraphics[width=0.33\linewidth, keepaspectratio]{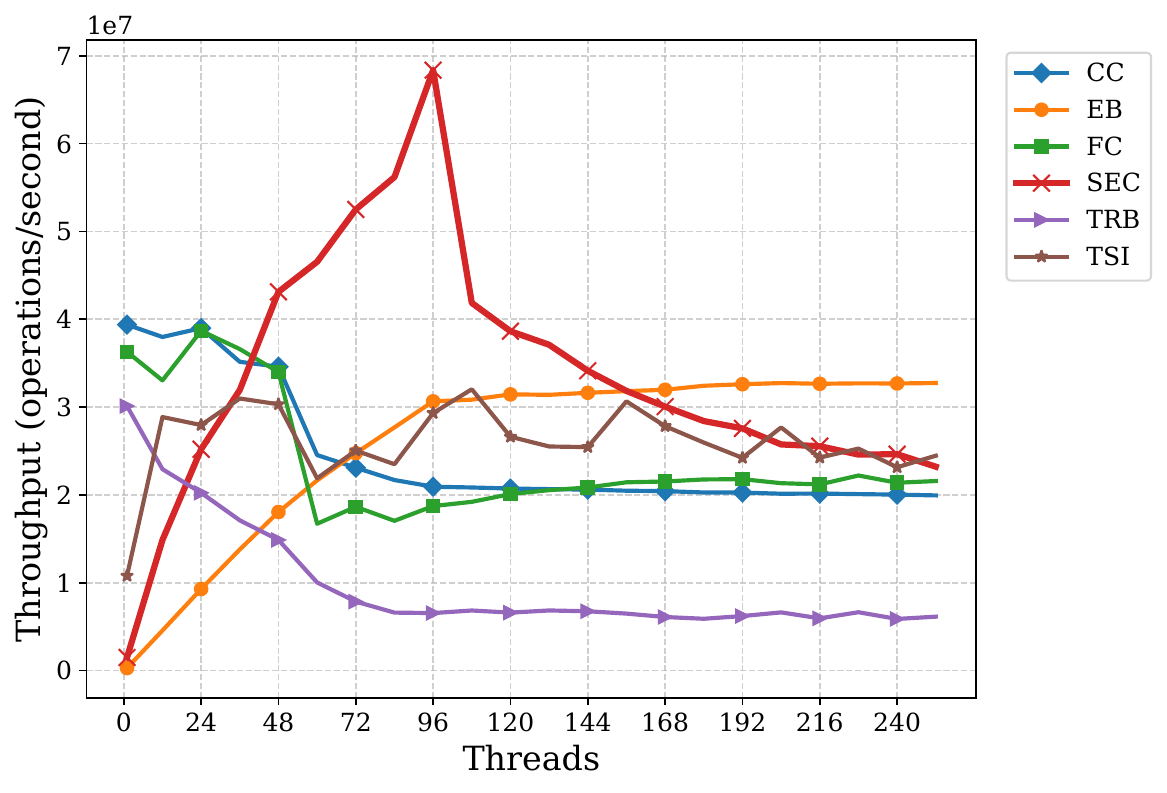}\hfill
    \includegraphics[width=0.33\linewidth, keepaspectratio]{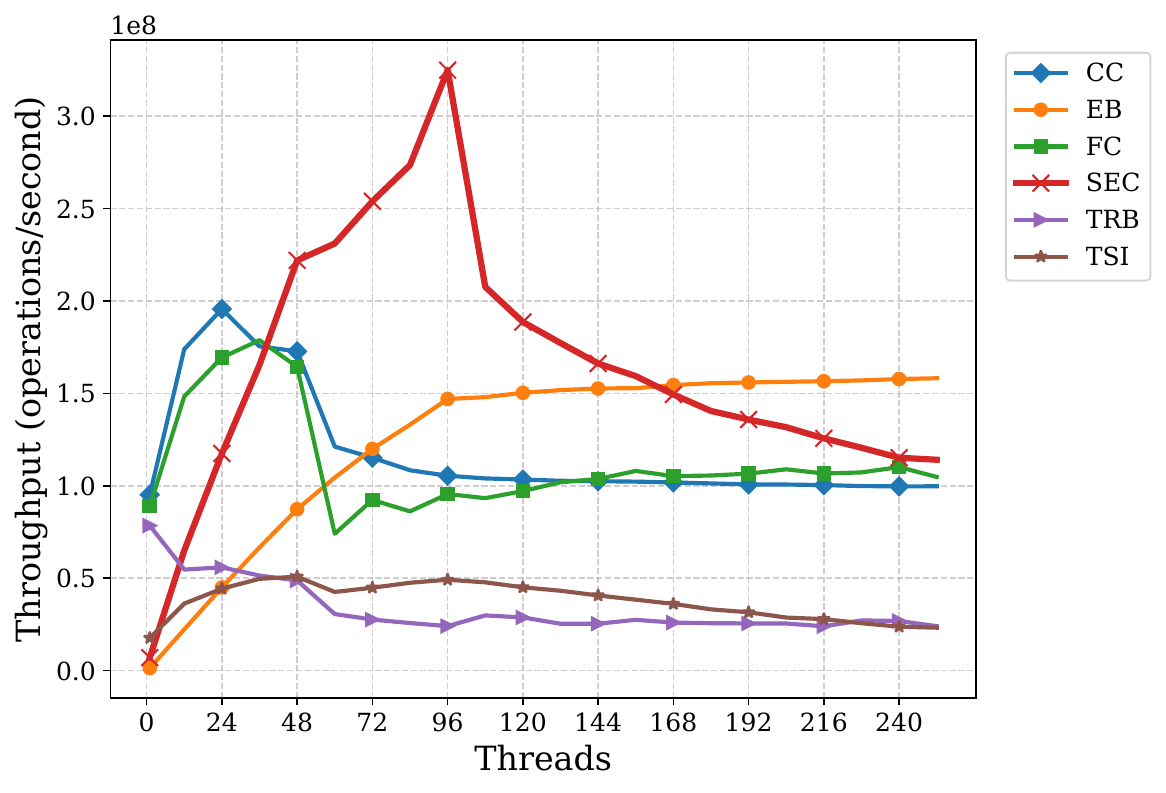}\hfill
    \caption{Throughput with varying threads. (Left) 100\% updates. (Middle) 50\% updates. (Right) 10\%updates. Y-axis: throughput in millions of operations per second. X-axis: \#threads.}
    \label{fig:exp1tit}
\end{figure}

\begin{figure}[!htb]
    \includegraphics[width=0.49\linewidth, keepaspectratio]{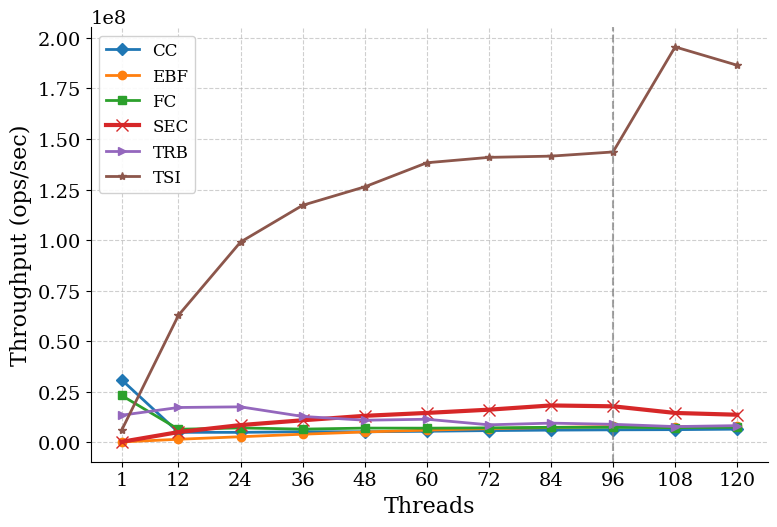}\hfill
    \includegraphics[width=0.49\linewidth, keepaspectratio]{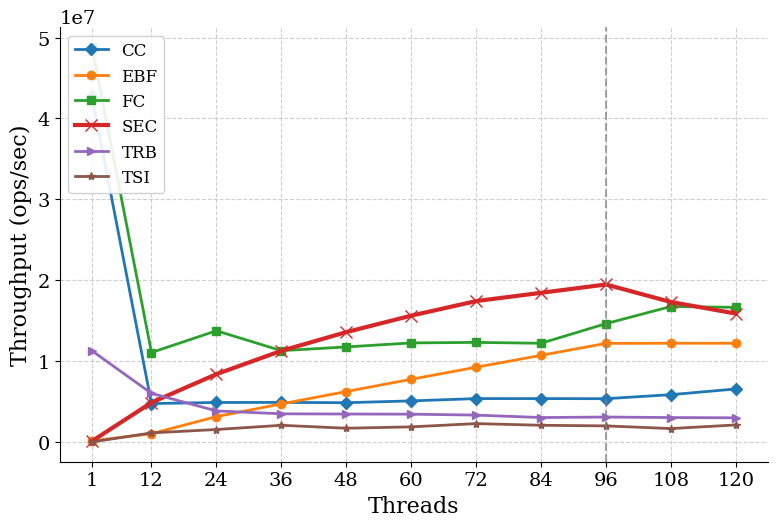}\hfill
    \caption{Throughput for push only and pop only workloads. (Left) Push only. (Right) Pop only. Y-axis: throughput in millions of operations per second. X-axis: \#threads. Number of aggregators used is two.
    }
    \label{fig:exp2tit}
\end{figure}

\begin{figure}[!htb]
\centering
    \includegraphics[width=0.33\linewidth, keepaspectratio]{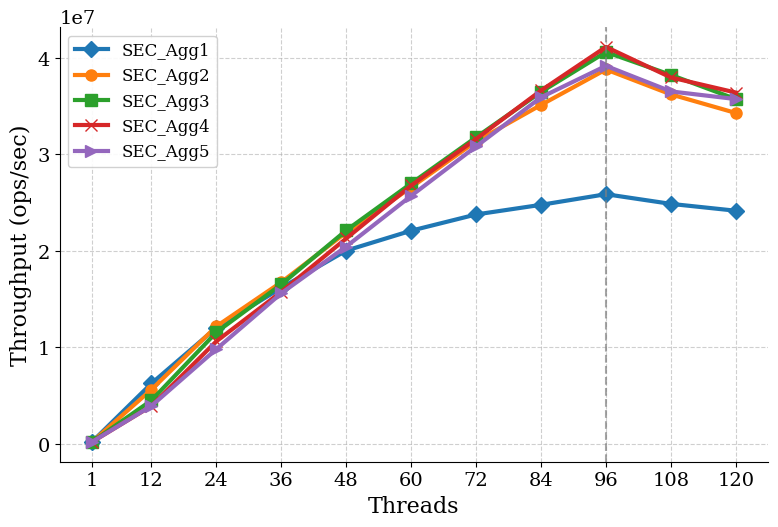}\hfill
    \includegraphics[width=0.33\linewidth, keepaspectratio]{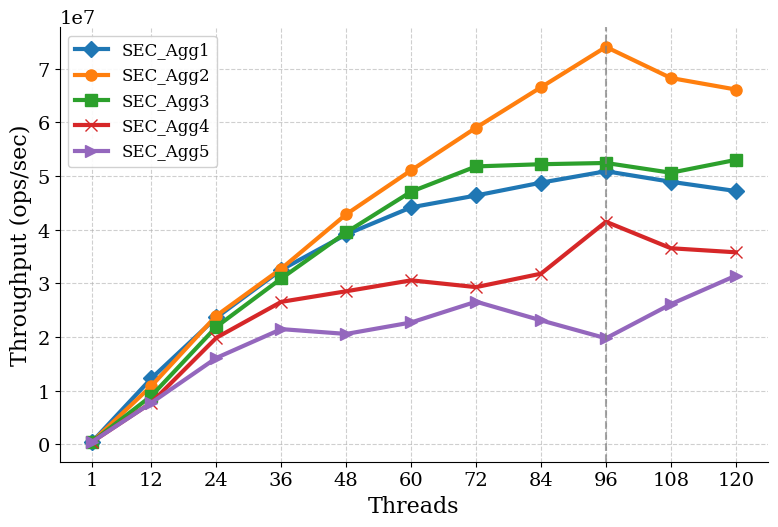}\hfill
    \includegraphics[width=0.33\linewidth, keepaspectratio]{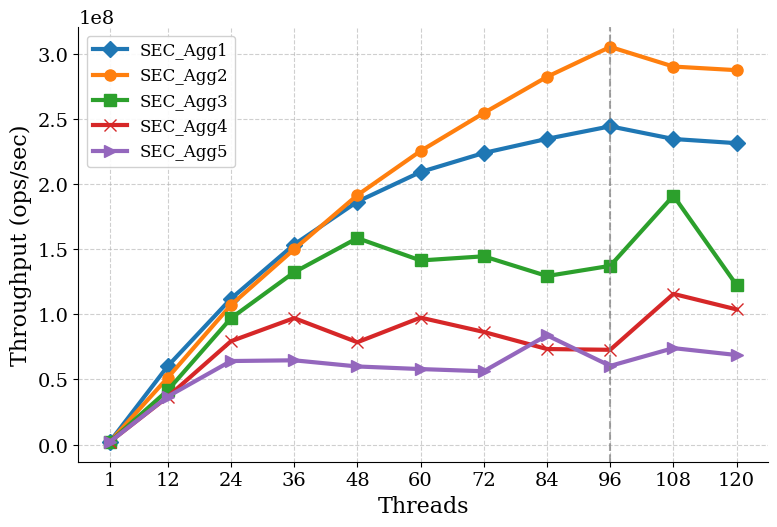}\hfill
    \caption{Self comparison with aggregators Throughput with varying threads. (Left) 100\% updates. (Middle) 50\% updates. (Right) 10\%updates. Y-axis: throughput in millions of operations per second. X-axis: \#threads.
    }
    \label{fig:exp3tit}
\end{figure}

\begin{figure}[!htb]
\centering
    \includegraphics[width=0.49\linewidth, keepaspectratio]{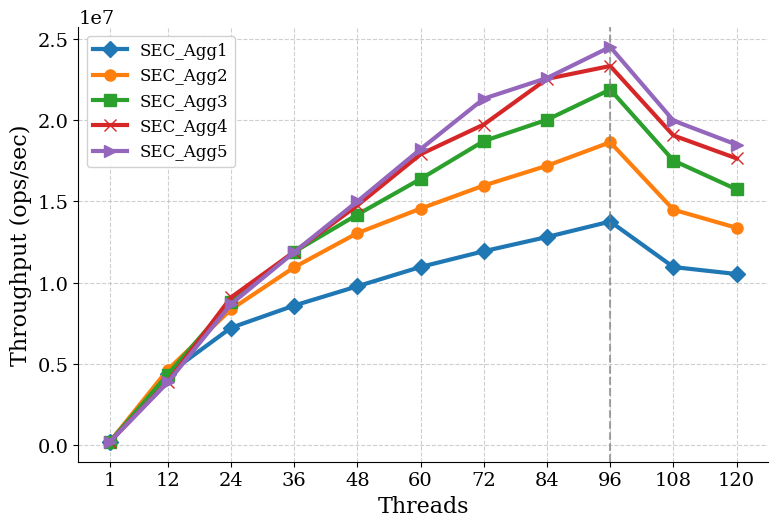}\hfill
    \includegraphics[width=0.49\linewidth, keepaspectratio]{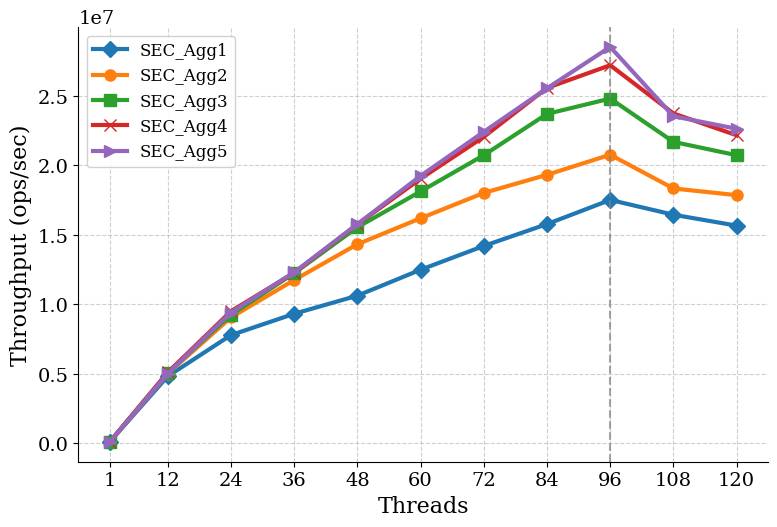}\hfill
    \caption{Self comparison with aggregators for push only and pop only workloads. Throughput with varying threads. (Left) Push only. (Right) Pop only. Y-axis: throughput in millions of operations per second. X-axis: \#threads.
    }
    \label{fig:exp3tit}
\end{figure}

\begin{table}[!htb]
\centering
\begin{tabular}{|c|c|c|c|}
\hline
      & \multicolumn{3}{c|}{\textbf{SEC}} \\ \hline
\textbf{Workload $\rightarrow$}      & \textbf{100\% upd} & \textbf{50\% upd} & \textbf{10\% upd} \\ \hline
Batching Degree    & 40    & 31   & 28  \\ \hline
\%Elimination      &  85\%   &  85\%  & 84\%  \\ \hline
\%Combining        & 15\%    &  15\%  & 16\%   \\ \hline
\end{tabular}
\caption{Batching degree (average size of batches during an execution), \%elimination (percent of operations eliminated per batch), and \%combining (percent of operations not eliminated per batch) in SEC for experiments in~\figref{exp1tit}. Column label 100 implies workloads with 100\% updates and so on.}
\label{tab:exp4tit}
\end{table}

\clearpage
\section{Experiments on 192 Threads Machine: Sapphire}
\label{sec:rapidsextexp}

We have also conducted experiments on an Intel Sapphire Rapids  (\textbf{Sapphire}) with 8 NUMA nodes supporting 2-way hyperthreading with 24 threads running on each NUMA node,
thus supporting 192 hardware threads in total. The machine has the following characteristics: 210MiB L3 cache, 3.5 GHz frequency, 380GB RAM.


\begin{figure}[!thb]
        \includegraphics[width=0.33\linewidth, keepaspectratio]{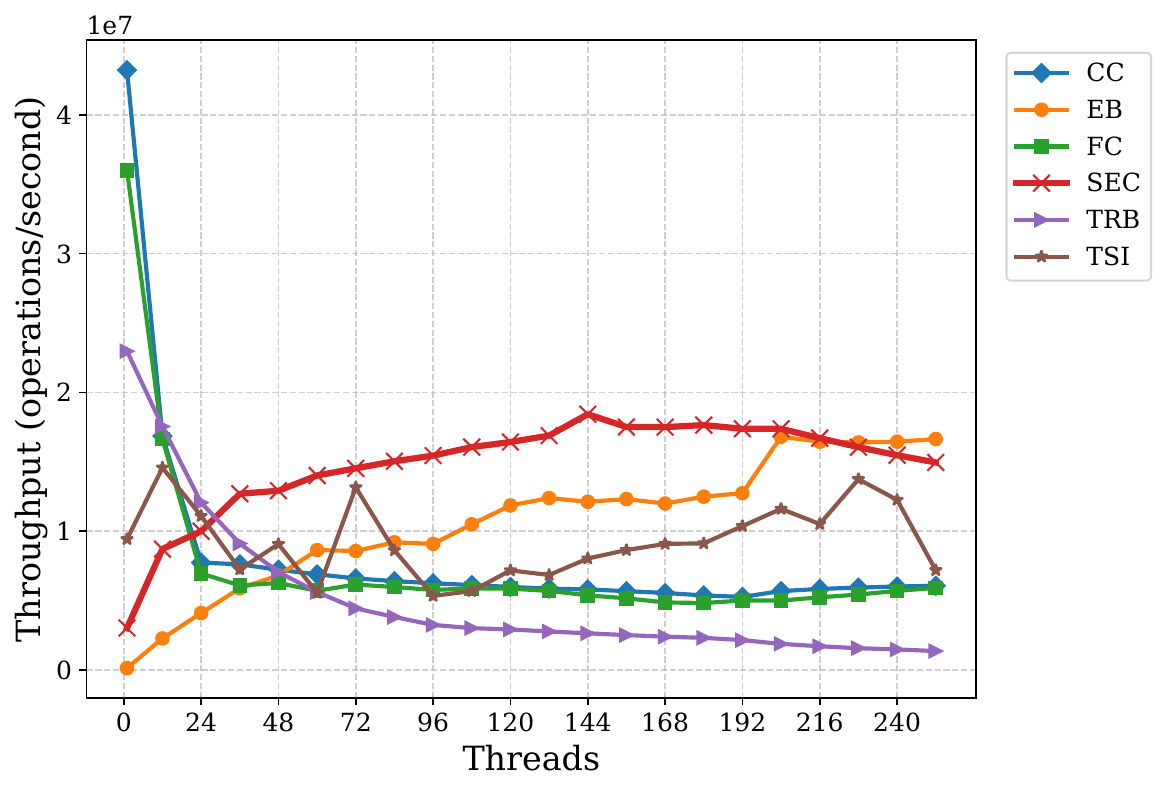}\hfill
        \includegraphics[width=0.33\linewidth, keepaspectratio]{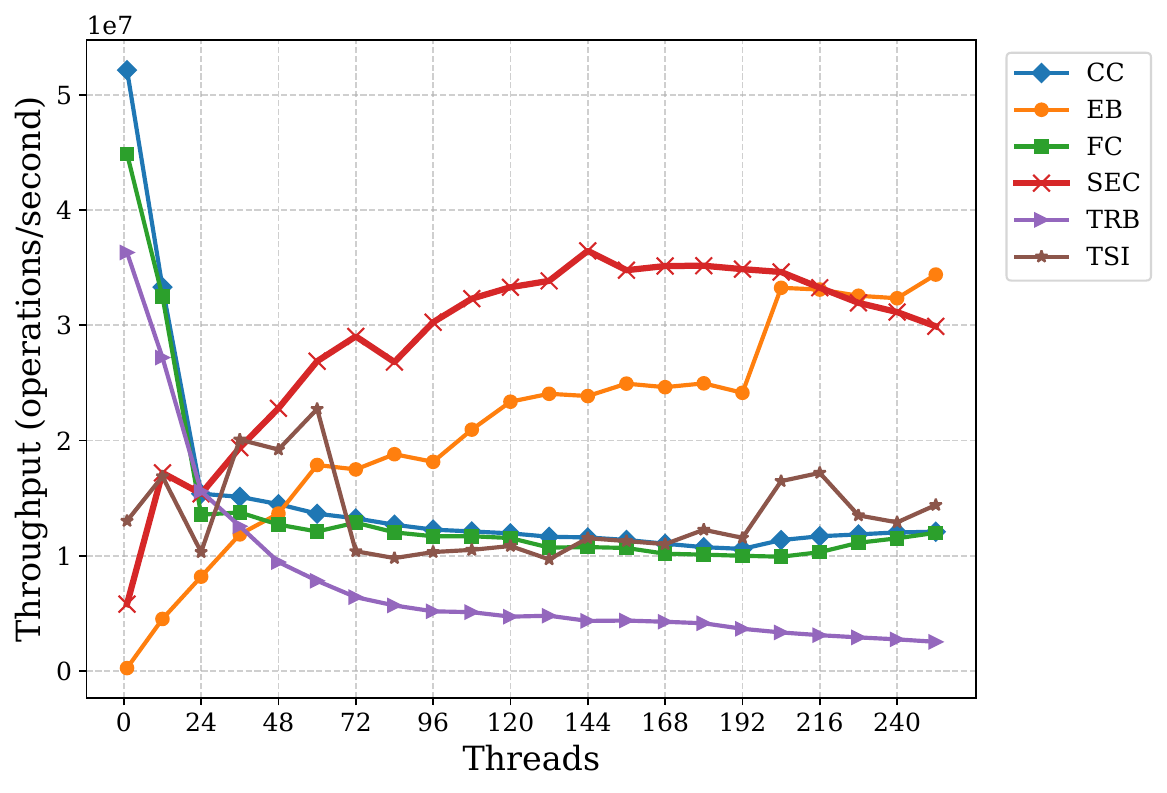}\hfill
        \includegraphics[width=0.33\linewidth, keepaspectratio]{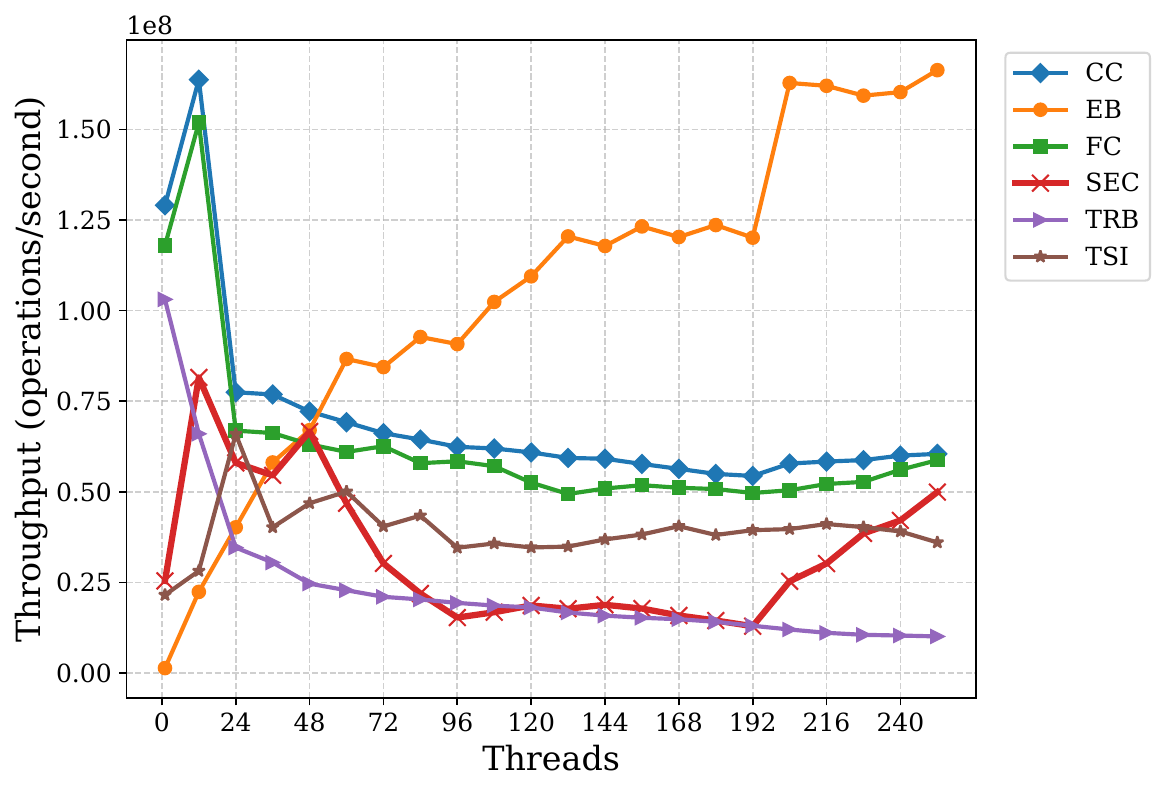}\hfill
    \caption{Throughput. (Left) 100\% updates. (Middle) 50\% updates. (Right) 10\%updates. Y-axis: throughput in millions of operations per second. X-axis: \#threads. Number of aggregators used is two.}
    \label{fig:exp1maxm}
\end{figure}

\begin{figure}[ht]
    \includegraphics[width=0.49\linewidth, keepaspectratio]{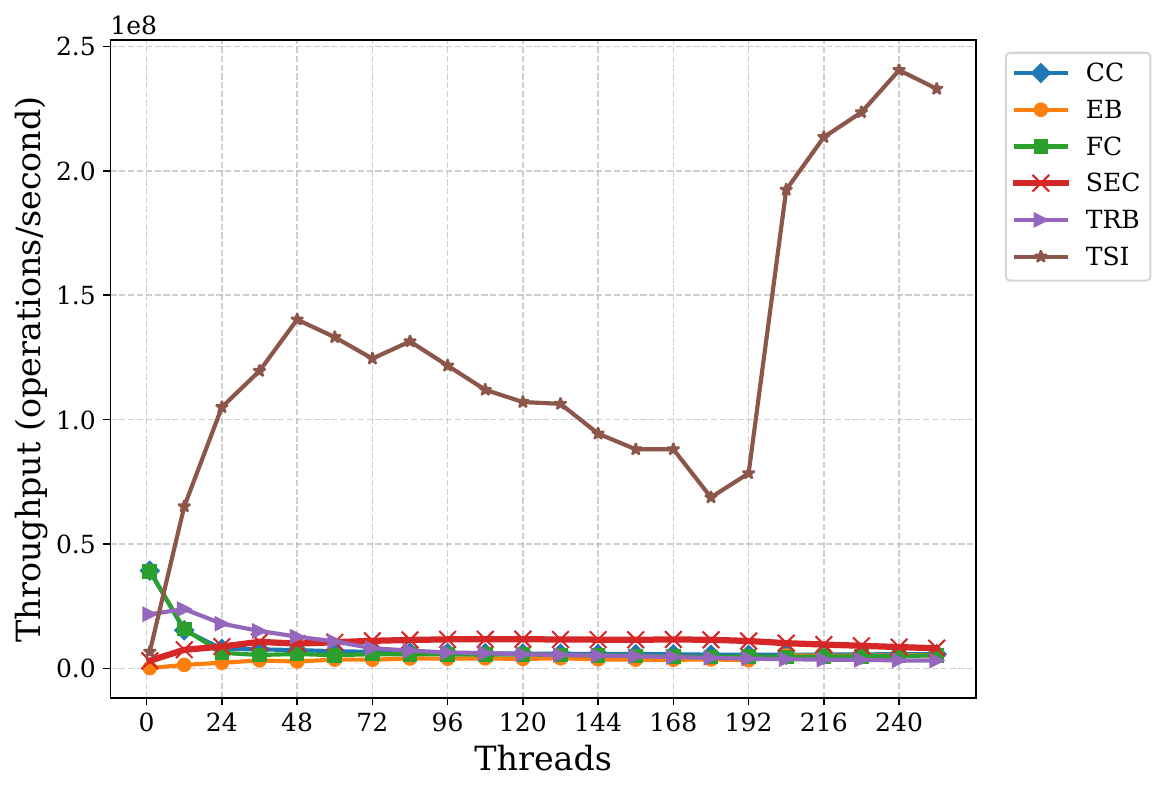}\hfill
    \includegraphics[width=0.49\linewidth, keepaspectratio]{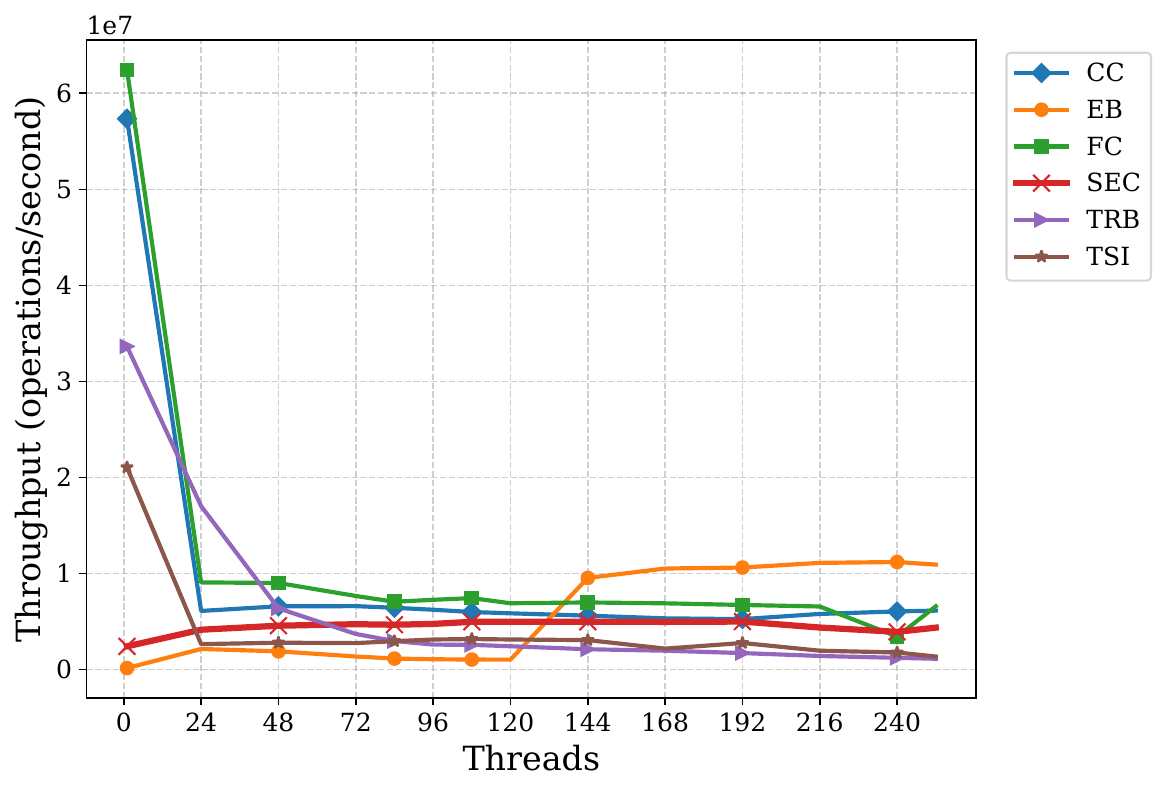}\hfill
    \caption{Throughput for push only and pop only workloads. (Left) Push only. (Right) Pop only. Y-axis: throughput in millions of operations per second. X-axis: \#threads. Number of aggregators used is two. }
    \label{fig:exp2maxm}
\end{figure}


\begin{figure}[ht]
\centering
        \includegraphics[width=0.33\linewidth, keepaspectratio]{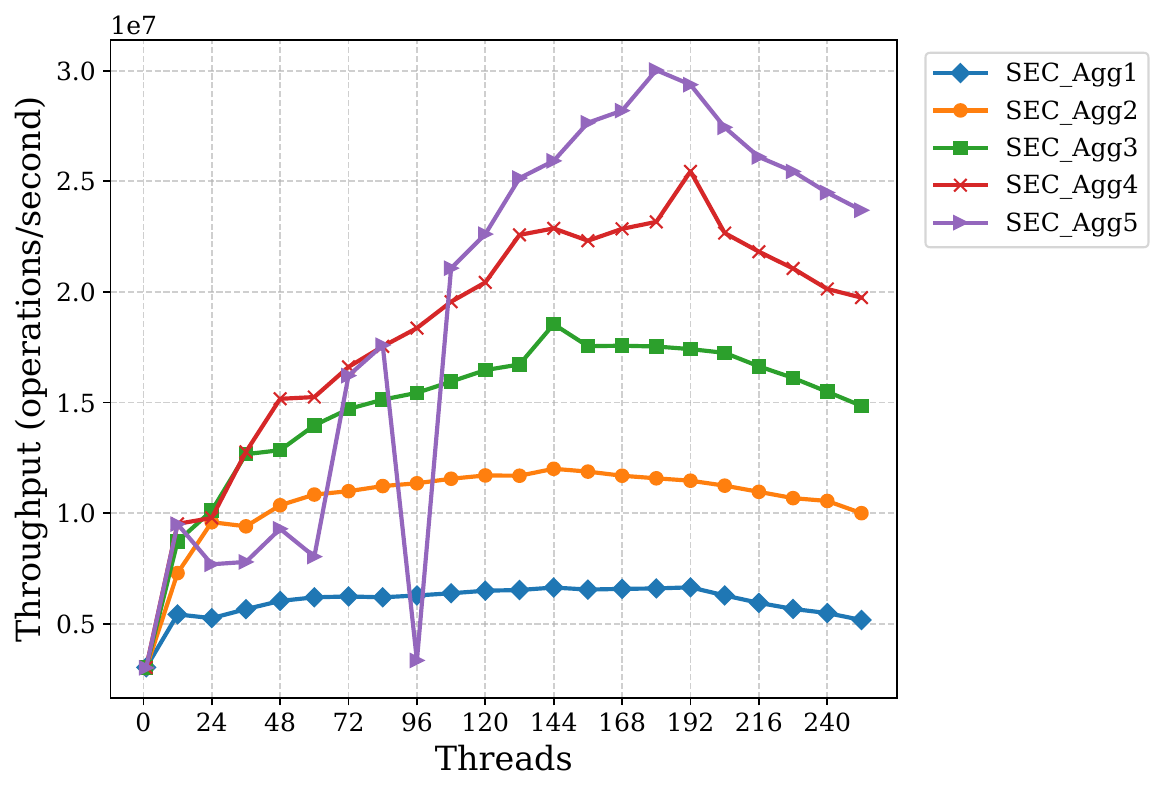}\hfill
        \includegraphics[width=0.33\linewidth, keepaspectratio]{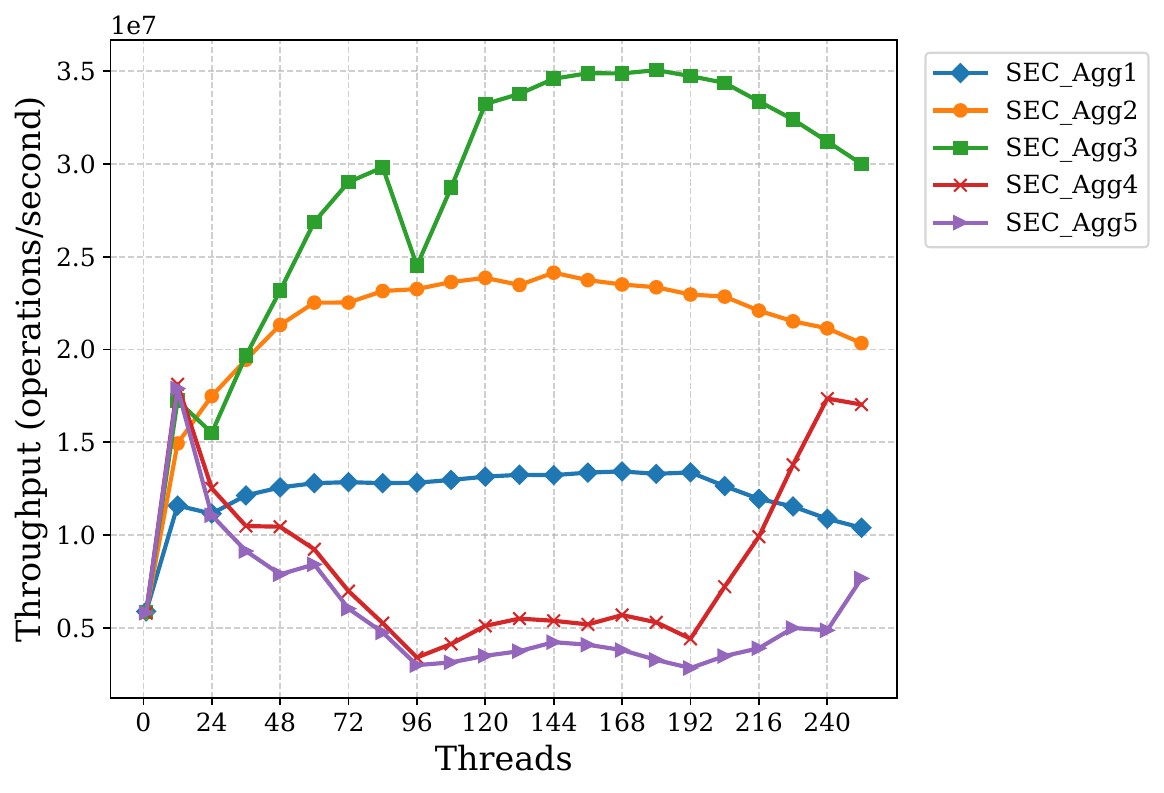}\hfill
        \includegraphics[width=0.33\linewidth, keepaspectratio]{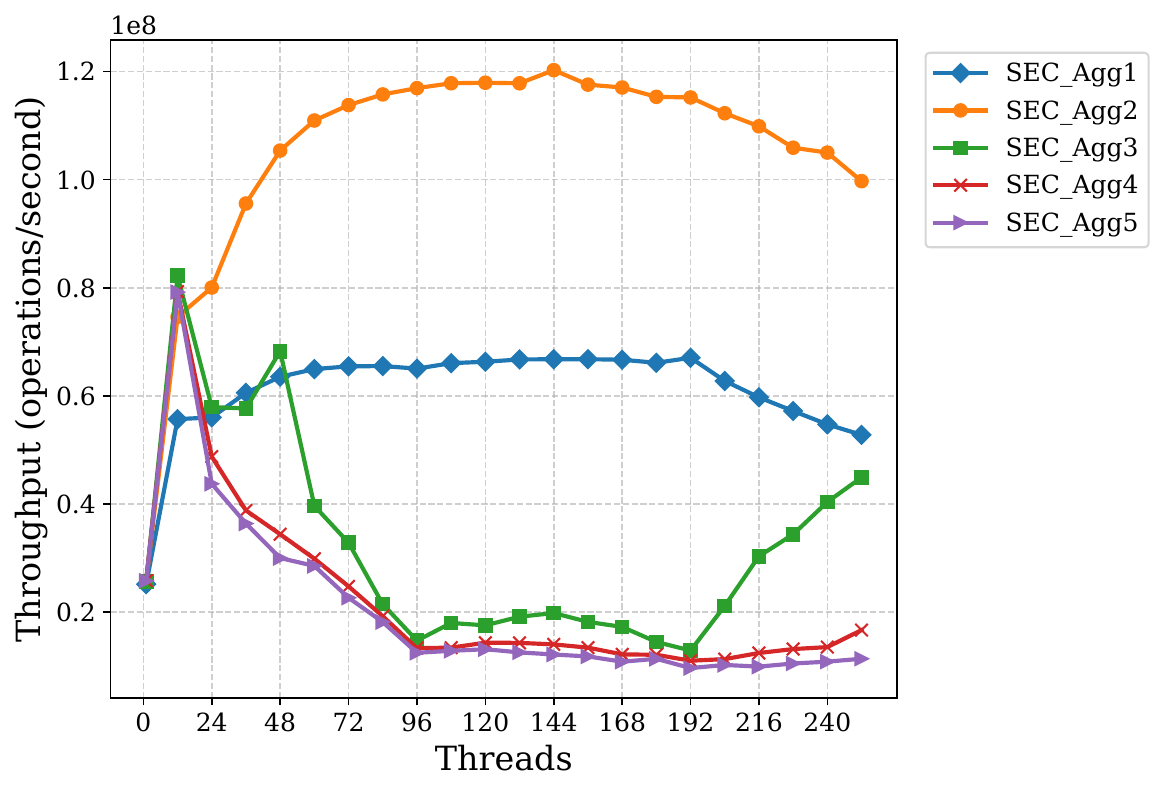}\hfill
    \caption{Comparing DCE+ throughput with various number of aggregators. From left to right, 100\% updates, 50\% updates, 10\%updates, 100\%push-only. Y-axis: Throughput. X-axis: \#threads. DCE+ with 1 aggregator is labeled as DCE+\_Agg1.}
    \label{fig:exp3maxm}
\end{figure}

\begin{figure}[ht]
\centering
        \includegraphics[width=0.49\linewidth, keepaspectratio]{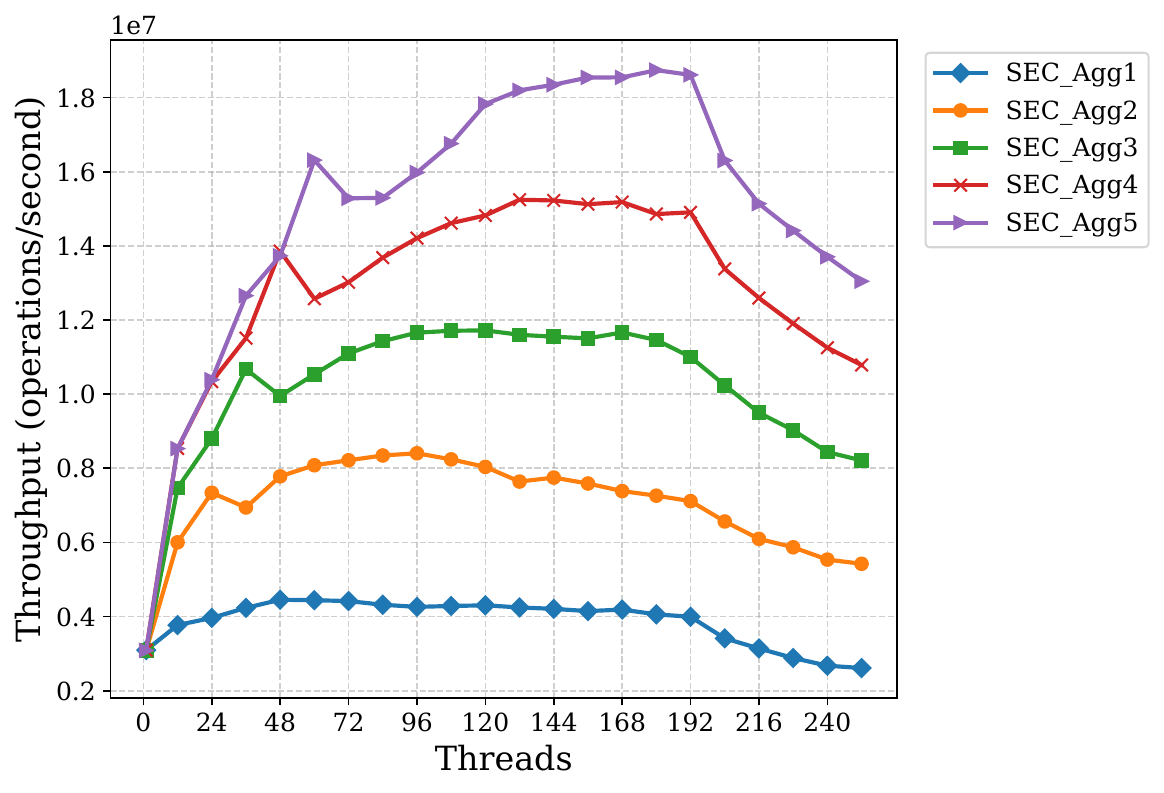}\hfill
        \includegraphics[width=0.49\linewidth, keepaspectratio]{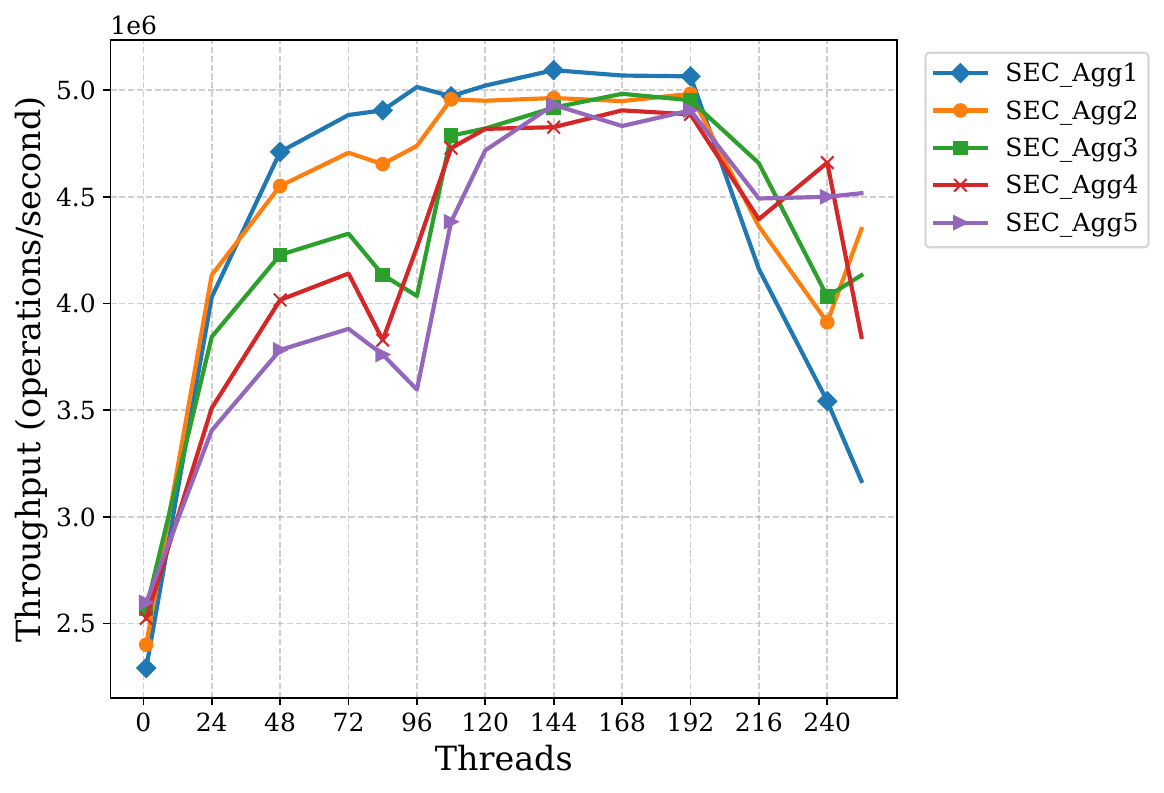}\hfill
    \caption{Self comparison with aggregators for push only and pop only workloads. Throughput with varying threads. (Left) Push only. (Right) Pop only. Y-axis: throughput in millions of operations per second. X-axis: \#threads.}
    \label{fig:exp3maxm}
\end{figure}

\begin{table}[!htb]
\centering
\begin{tabular}{|c|c|c|c|}
\hline
      & \multicolumn{3}{c|}{\textbf{SEC}} \\ \hline
\textbf{Workload $\rightarrow$}      & \textbf{100\% upd} & \textbf{50\% upd} & \textbf{10\% upd} \\ \hline
Batching Degree    & 24    & 22   & 17  \\ \hline
\%Elimination      &  77\%   &  75\%  & 70\%  \\ \hline
\%Combining        & 23\%    &  25\%  & 30\%   \\ \hline

\end{tabular}
\caption{Batching degree (average size of batches during an execution), \%elimination (percent of operations eliminated per batch), and \%combining (percent of operations not eliminated per batch) in SEC for EXP1. Column label 100 implies workloads with 100\% updates and so on.}
\label{tab:exp4maxm}
\end{table}

